\def\isarxiv{1}   % If de-comment this line, then it is the arxiv version.
\ifdefined\isarxiv
\documentclass[11pt]{article}
\usepackage[margin=1in]{geometry}

\else
\documentclass{article}
\usepackage{ltexpprt}
\fi

\usepackage[utf8]{inputenc}
\usepackage{amsmath,amssymb,amsthm}
\usepackage{hyperref}
\hypersetup{colorlinks,allcolors=[rgb]{0,0.13672,0.95703}}
\usepackage{cleveref,xcolor}
\usepackage[affil-sl]{authblk}
\usepackage{graphicx}
\usepackage{cite}
\usepackage{diagbox}
\usepackage{cancel}
\usepackage{algorithm}
\usepackage{algpseudocode}
\usepackage{soul} % for comments
\usepackage{tikz,pgfplots}
\usetikzlibrary{patterns} % Load the patterns library
\usepackage{mathtools}
\usepackage{enumitem}
\newcommand\T{{*}}

\newcommand\tr{{\operatorname{tr}}}

\newcommand\E{{\mathbb{E}}}
\newcommand\C{{\mathbb{C}}}

\renewcommand\v[1]{{\boldsymbol{#1}}}

\newcommand{\R}[0]{\mathbb{R}}

\newtheorem{definition}{Definition}[section]
\newtheorem{problem}{Problem}[section]

\newtheorem{theorem}{Theorem}[section]

\newtheorem{fact}{Fact}[section]
\newtheorem{lemma}{Lemma}[section]
\newtheorem{corollary}{Corollary}[section]

\newtheorem{remark}{Remark}[section]
\theoremstyle{definition}

\theoremstyle{remark}

\numberwithin{equation}{section}

\renewcommand{\b}{\mathbf{b}}
\newcommand{\A}{\mathbf{A}}
\newcommand{\K}{\mathbf{K}}
\newcommand{\B}{\mathbf{B}}

\newcommand{\U}{\mathbf{U}}
\newcommand{\V}{\mathbf{V}}
\newcommand{\W}{\mathbf{W}}
\newcommand{\M}{\mathbf{M}}
\newcommand{\x}{\mathbf{x}}
\newcommand{\y}{\mathbf{y}}
\newcommand{\z}{\mathbf{z}}
\newcommand{\w}{\mathbf{w}}
\newcommand{\q}{\mathbf{q}}
\newcommand{\e}{\mathbf{e}}
\newcommand{\g}{\mathbf{g}}

\newcommand{\Q}{\mathbf{Q}}
\newcommand{\eps}{\epsilon}

\newcommand{\I}{\mathbf{I}}

\renewcommand{\r}{\mathbf{r}}
\renewcommand{\T}{\mathbf{T}}
\newcommand{\bpi}{\mathbf{\Pi}}
\newcommand{\Sig}{\mathbf{\Sigma}}

\newcommand{\nys}{\mathrm{nys}}
\newcommand{\nnz}{\mathrm{nnz}}

\newcommand{\poly}{\mathrm{poly}}
\newcommand{\polylog}{\mathrm{polylog}}

\renewcommand{\u}{\mathbf{u}}
\renewcommand{\v}{\mathbf{v}}

\renewcommand{\H}{\mathbf{H}}

\renewcommand{\S}{\mathbf{S}}
\renewcommand{\C}{\mathbf{C}}
\renewcommand{\P}{\mathbf{P}}
\newcommand{\Z}{\mathbf{Z}}

\renewcommand{\c}{\mathbf{c}}

\def\AlgFullName{Multi-level Sketched Preconditioning}
\def\AlgName{MSP}

\DeclareMathOperator*{\mach}{\epsilon_{mach}}

\ifdefined\isarxiv
\title{Faster Linear Systems and Matrix Norm Approximation via Multi-level Sketched Preconditioning}

\date{}
\author{Micha{\l} Derezi\'nski\thanks{University of Michigan (\texttt{derezin@umich.edu})}
\quad
Christopher Musco\thanks{New York University (\texttt{cmusco@nyu.edu})}
\quad
Jiaming Yang\thanks{University of Michigan (\texttt{jiamyang@umich.edu})}}

\else
\title{Faster Linear Systems and Matrix Norm Approximation via Multi-level Sketched Preconditioning}

\author{Micha{\l} Derezi\'nski}
\affil{University of Michigan}
\author{Christopher Musco}
\affil{New York University}
\author{Jiaming Yang}
\affil{University of Michigan}
\fi

\begin{document}
\maketitle
\thispagestyle{empty}

\begin{abstract}
    We present a new class of preconditioned iterative methods for solving linear systems of the form $\A\x = \b$.
    Our methods are based on constructing a low-rank Nystr\"om approximation to $\A$ using sparse random matrix sketching. This approximation is used to construct a preconditioner, which itself is inverted quickly using additional levels of random sketching and preconditioning.

    We prove that the convergence of our methods depends on a natural \emph{average condition number} of $\A$, which improves as the rank of the Nystr\"om approximation increases. Concretely, this allows us to obtain faster runtimes for a number of fundamental linear algebraic problems:
    \begin{enumerate}[]
        \item We show how to solve any $n\times n$ linear system that is well-conditioned except for $k$ outlying large singular values in $\tilde{O}(n^{2.065} + k^\omega)$ time, improving on a recent result of [Derezi{\'n}ski, Yang, STOC 2024] for all $k \gtrsim n^{0.78}$.
        \item We give the first $\tilde{O}(n^2 + {d_\lambda}^{\omega}$) time algorithm for solving a regularized linear system $(\A + \lambda \I)\x = \b$, where $\A$ is positive semidefinite with effective dimension $d_\lambda=\tr(\A(\A+\lambda\I)^{-1})$. This problem arises in applications like Gaussian process regression.
        \item We give faster algorithms for approximating Schatten $p$-norms and other matrix norms. For example, for the Schatten 1-norm  (nuclear norm), we give an algorithm that runs in $\tilde{O}(n^{2.11})$ time, improving on an $\tilde{O}(n^{2.18})$ method of [Musco et al., ITCS 2018].
    \end{enumerate}
    All results are proven in the real RAM model of computation. 
Interestingly, previous state-of-the-art algorithms for most of the problems above relied on stochastic iterative methods, like stochastic coordinate and gradient descent.
Our work takes a completely different approach, instead leveraging tools from matrix sketching.
\end{abstract}

\newpage
\setcounter{page}{1}

\section{Introduction}\label{s:intro}
We consider the complexity of solving a system of linear equations: given an $n\times n$ matrix $\A$ and an $n$-dimensional vector $\b$, find $\x$ such that $\A\x=\b$. This ubiquitous task in numerical linear algebra has applications across data science and machine learning, the physical sciences, engineering, and more. Despite many existing methods for solving linear systems, they remain a major computational bottleneck, so faster algorithms are a subject of active research. Much of the progress on faster algorithms  concentrates on designing methods tailored for matrices with special structure, such as Laplacian, Toeplitz, or Hankel matrices, among others \cite{kailath1979displacement,xia2012superfast,spielman2014nearly,koutis2012fast,kyng_sachdeva_2016,cohen2018solving,PengVempala:2021}. 
Progress on general, unstructured linear systems has been slower.

One direction for developing faster algorithms is to improve on fast matrix multiplication. Thanks to Strassen's reduction showing that matrix inversion is equivalent to matrix multiplication \cite{strassen1969gaussian,pan1984multiply,coppersmith1987matrix,williams2012multiplying}, general linear systems can be solved in $O(n^\omega)$ time, where $\omega<2.372$ is the current best known matrix multiplication exponent \cite{williams2023new}. Another approach is to solve linear systems via iterative refinement, for instance using deterministic methods such as the Conjugate Gradient (CG) or Lanczos algorithms \cite{hestenes1952methods}, or stochastic approaches like randomized coordinate descent \cite{strohmer2009randomized,leventhal2010randomized,lee2013efficient}. Iterative methods tend to be faster in many practical settings, as their runtime typically scales only quadratically with $n$, i.e., as $O(n^2)$. However, the runtime of typical iterative methods involves a multiplicative factor depending on the condition number of $\A$ (the ratio between its largest and smallest singular values) or related parameters, making them largely incomparable to solvers based on fast matrix multiplication. 

One way to address this issue is  to introduce a problem parameter $k$, such that any ill-conditioned part of $\A$ is restricted to a $k$-dimensional subspace. Formally, we ask: 
\begin{problem}\label{p:large-sv}
    What is the time complexity of solving an $n\times n$ linear system $\A\x=\b$ such that matrix $\A$ has at most $k$ singular values larger than $O(1)$ times its smallest singular value?
\end{problem}
When $k = n$, our fastest methods for solving \Cref{p:large-sv} run in $O(n^\omega)$ time via fast matrix multiplication. When $k=0$, the problem can be solved in $\tilde O(n^2)$ time via any standard iterative solver such as CG. We are interested in the arguably more interesting intermediate regime, where $\A$ is a combination of a low-rank ill-conditioned matrix and a full-rank well-conditioned one. Such matrices arise often in practice,  either due to various forms of regularization that are prevalent in machine learning, statistics and optimization \cite{boyd2004convex,zhang2013divide,dobriban2018high,derezinski2024recent}, or due to the presence of isotropic noise coming from measurement error, rounding, or compression \cite{loh2011high,liang2021pruning}.

A standard approach to Problem \ref{p:large-sv} is to construct an approximation to the rank $k$ subspace identified by $\A$'s top singular values. That information can be used to precondition an iterative solver so that it runs in $\tilde O(n^2)$ time \cite{halko2011finding}. Constructing a sufficiently accurate approximation to the top subspace requires (block) power iteration or related methods \cite{musco2015randomized}, which take $\tilde{O}(n^{\omega(1,1,\log_n k)})$ time, where $\omega(1,1,\log_nk)\in[2,\omega]$ is the exponent of rectangular matrix multiplication between an $n\times n$ and $n\times k$ matrix \cite{Le-Gall:2012}.\footnote{This corresponds to the $O(n^2k)$ time for classical multiplication of $n\times n$ and $n\times k$ matrices.} This approach already interpolates between the two extremes of the problem. In fact, it achieves a near-optimal $\tilde O(n^2)$ runtime for $k=O(n^{0.32})$.

Recently, \cite{derezinski2023solving} give the first improvement on this baseline by presenting a randomized Kaczmarz-like iterative method that solves \Cref{p:large-sv} in time $\tilde O(n^2+nk^{\omega-1})$. This result leads to an optimal $\tilde O(n^2)$ runtime for any $k=O(n^{\frac1{\omega-1}})=O(n^{0.73})$, and for larger $k$ it gradually degrades to $O(n^\omega)$.

\subsection{Main Results}
It might appear as though the result of  \cite{derezinski2023solving} is optimal. In particular, the currently best known algorithms for finding even a coarse (Frobenius norm error) approximation of the top rank $k$ subspace of $\A$ also require $\tilde O(n^2+nk^{\omega-1})$ time \cite{clarkson2013low,cohen2015dimensionality,cohen2017input,chepurko2022near}, and finding such an approximation seems like a prerequiste for solving \Cref{p:large-sv}. 
The contribution of this paper is to present an algorithm that breaks through this complexity barrier. Along the way, we give new state-of-the-art runtimes for two other central problems in linear algebra that are closely related to \Cref{p:large-sv}: kernel ridge regression \cite{acw17} and Schatten norm approximation \cite{musco2018spectrum}.

Concretely, we prove the following result, which  improves on the previously best known time complexity of $\tilde O(n^2+nk^{\omega-1})$ for \Cref{p:large-sv} for all $k=\Omega(n^{0.78})$. See Figure \ref{fig:complexity} for a comparison. 
\begin{theorem}[Main result, informal Theorem \ref{thm:main_rec}]\label{thm:main} 
Given an invertible $n\times n$ matrix $\A$ with at most $k$ singular values larger than $O(1)$ times its smallest singular value, and a length $n$ vector $\b$, there is an algorithm that, with high probability, computes $\tilde\x$ such that $\|\A\tilde\x-\b\|\leq\epsilon\|\b\|$ in time\footnote{We use $\tilde{O}(\cdot)$ to hide logarithmic dependencies on the  dimension $n$ and the condition number of $\A$.}:
\begin{align*}
\tilde O\left(n^{2.065} \cdot \log^3 1/\epsilon + k^\omega\right). 
\end{align*}
\end{theorem}

\begin{figure}[ht]
\centering
\begin{tikzpicture}[scale=2.5,font=\footnotesize]
\tikzset{%
    dot/.style={circle, draw, fill=black, inner sep=0pt, minimum
    width=3pt},
    ldot/.style={circle, draw, black, inner sep=0pt, minimum
    width=2pt}, 
    top/.style={anchor=south, inner sep=3pt},
    bottom/.style={anchor=north, inner sep=3pt},
}
  \draw[->] (0,2) -- (2.2,2) node[right] {$\theta$};
  \draw[->] (0,2) -- (0,4.05) node[above] {$\beta$};

  \fill[fill=gray!30] (1.458,2) -- (1.686,2) -- (1.741,2.32) -- (1.552,2.32) -- cycle;
  \fill[fill=red!30] (1.686,2) -- (2,2) -- (2,3.86) -- cycle;
  \fill[fill=yellow!30] (1.741,2.32) -- (1.552,2.32) -- (2,3.86) -- cycle;
  \node[ldot] (O) at (0,2) {};
  \node[ldot] (A) at (0.642,2) {};
  \node[ldot] (B) at (1.458,2) {};
  \node[ldot] (C) at (1.686,2) {};
  \node[ldot] (D) at (1.741,2.32) {};
  \node[ldot] (Dy) at (0,2.32) {};
  \node[ldot] (E) at (1.552,2.32) {};
  \node[ldot] (T) at (2,3.86) {};
  \node[ldot] (Tx) at (2,2) {};
  \node[ldot] (Ty) at (0,3.86) {};

  \draw [dotted] (T) -- (Tx);
  \draw [dotted] (T) -- (Ty);

  \node[bottom] at (O) {$0$};
  \node[bottom] at (A) {$.32$}; 
  \node[bottom] at (B) {$.73$}; 
  \node[bottom] at (C) {$.84$}; 
  \node[bottom] at (Tx) {$1$}; 
  \node[left] at (O) {$2$};
  \node[left] at (Dy) {$2.065$}; 
  \node[left] at (Ty) {$2.372$}; 
  \draw[-][color=gray, very thick] (O) -- (A);
  \draw[-][color=gray, thick, dashed] (A) -- (T);
  \draw[-][color=blue, very thick] (A) -- (B);
  \draw[-][color=blue, very thick] (B) -- (E);
  \draw[-][color=blue, thick,dashed] (E) -- (T);
  \draw[-][color=yellow, thick, dashed] (Dy) -- (E);
  \draw[-][color=yellow, very thick] (E) -- (D);
  \draw[-][color=yellow, very thick] (D) -- (T);
  \draw[-][color=red, very thick] (B) -- (C);
  \draw[-][color=red, very thick] (C) -- (D);

  \begin{scope}[shift={(2.3,3.5)}]
    \draw[color=gray, thick] (0,0) -- (0.2,0) node[right, black, xshift=2pt] {Power Method-based Preconditioning \cite{halko2011finding}};
    \draw[color=blue, thick] (0,-0.25) -- (0.2,-0.25) node[right, black, xshift=2pt] {Randomized Block Kaczmarz \cite{derezinski2023solving}};
    \draw[color=yellow, thick] (0,-0.5) -- (0.2,-0.5) node[right, black, xshift=2pt] {\AlgFullName{} (\textbf{Thm \ref{thm:main}})};
    \draw[color=red, thick] (0,-0.75) -- (0.2,-0.75) node[right, black, xshift=2pt] {Conditional lower bound (\textbf{Thm \ref{t:lower}})};
  \end{scope}
\end{tikzpicture}
\caption{Time complexity for solving an $n\times n$ linear system with $k = n^\theta$ large singular values under current matrix multiplication exponent $\omega \approx 2.372$. The $x$-axis denotes the exponent $\theta$, while the $y$-axis denotes the exponent $\beta$ in the time complexity $\tilde{O}(n^\beta)$. The yellow line is our work (Theorem~\ref{thm:main}), with the yellow area showing the complexity improvement compared to prior work. The red line denotes a lower bound for the problem, which we prove in Theorem \ref{t:lower}. The red area is unachievable under the assumption that solving general dense linear systems requires $\Omega(n^\omega)$~time.}\label{fig:complexity}
\end{figure}
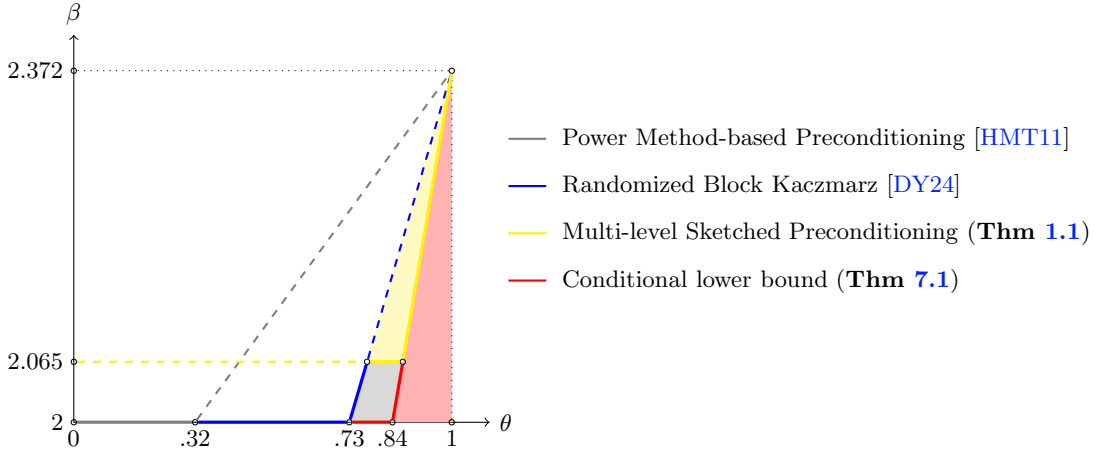

To understand the significance of \Cref{thm:main}, consider the following special instance of Problem~\ref{p:large-sv}, where $\A$ is a block-diagonal matrix with one $k\times k$ block that contains an arbitrary ill-conditioned matrix, and a second $(n\!-\!k)\times (n\!-\!k)$ block that contains a well-conditioned (but non-trivial) matrix. Solving such a linear system is equivalent to solving a small worst-case $k\times k$ linear system and a dense, well-condition  system. Thus, under the assumption that the worst-case time complexity of linear systems is governed by the complexity of matrix multiplication, we can lower bound the cost of solving Problem \ref{p:large-sv} with current fast matrix multiplication by $\Omega(n^2+k^\omega)$ (this is formalized in Theorem \ref{t:lower}). In comparison, our approach yields $\tilde O(n^{2.065}+k^\omega)$ time, which matches the lower bound up to $n^{0.065}$ everywhere and up to logarithmic factors for $k=\Omega(n^{0.87})$.

Surprisingly, our algorithm departs from the stochastic optimization approach used in the prior work of \cite{derezinski2023solving}. Instead, we employ a deterministic solver with a randomized preconditioner. To do so, we combine two main ingredients. First, building on prior work \cite{frangella2023randomized,epperly_random_precond}, we show that even a very coarse approximation to $\A$'s top rank $k$ subspace can be used to construct a ``good enough'' preconditioner. In particular, we can replace the use of algorithms like block power method with more efficient, but less accurate, linear-time sketching methods for low-rank approximation \cite{clarkson2013low}. Second, we show that it is possible to avoid the final step of such methods, which requires a prohibitive $\tilde O(nk^{\omega-1})$ cost to explicitly {construct} a rank $k$ subspace approximating $\A$'s top singular vectors. Instead, using $\tilde O(nk + k^\omega)$ preprocessing time, we show how to maintain a data structure that lets us perform inexact matrix-vector products with an inverted low-rank preconditioner. The data structure involves additional levels of solving linear systems with randomized sketching and preconditioning methods, resulting in a framework that we call \textbf{\AlgFullName{} (\AlgName)}. 
We discuss details of our approach further in Section \ref{s:techniques}. Before doing so, we highlight two additional applications of the framework beyond \Cref{thm:main}.

\paragraph{Regularized linear systems.}
An important case of linear systems that are well-conditioned except for a few large singular values are those that are explicitly regularized by adding a scaled identity, $\lambda \I$, to a positive semidefinite matrix $\A$. This setting arises, for instance, in kernel ridge regression \cite{em14,rcr17,musco2017recursive,acw17,frangella2023randomized} and second-order optimization \cite{levenberg1944method,marquardt1963algorithm,monteiro2012iteration,jiang2023online}. 
Prior results have shown that such linear systems can be solved in $\tilde O(n^2+n{d_\lambda}^{\omega-1})$ time, where $d_\lambda=\tr(\A(\A+\lambda\I)^{-1})\leq n$ is the so-called $\lambda$-effective dimension of the problem, often much smaller than $n$ \cite{acw17,musco2017recursive}. We improve on this by using \AlgFullName{}, leading to a time complexity of $\tilde O(n^2+{d_\lambda}^\omega)$, which is optimal (up to log factors).

\begin{theorem}[Regularized linear systems, informal Theorem \ref{thm:krr_formal}]\label{thm:krr}
Given an $n\times n$ positive semidefinite matrix $\A$, an $n$-dimensional vector $\b$, and $\lambda>0$, there is an algorithm that, with high probability, computes $\tilde\x$ such that $\|(\A+\lambda\I)\tilde\x - \b\|\leq \epsilon\|\b\|$ in time:
\begin{align*} 
    \tilde O\left(n^2 \cdot \log^3 1/\epsilon + {d_\lambda}^\omega \right),\quad\text{where}\quad d_\lambda = \tr(\A(\A+\lambda\I)^{-1}).
\end{align*}
\end{theorem}

\paragraph{Matrix norm estimation.} 
Linear systems often arise in algorithms for solving other matrix problems, including linear and semidefinite programming \cite{cohen2021solving,jiang2020faster}, least squares and $l_p$ regression \cite{sarlos2006improved,rokhlin2008fast,meng2013low,cohen2015lp,cd21}, and for estimating various matrix  properties \cite{musco2018spectrum}. Our new algorithms can be used to speed up any of these tasks. As a motivating example, we show how to improve methods for estimating various matrix norms, a fundamental problem in linear algebra. We build on an approach of \cite{musco2018spectrum}, which uses regularized linear system solves to build rational function approximations that, when combined with stochastic trace estimation methods \cite{Hutchinson:1990}, allow for the estimation
of various functions of the singular values faster than $O(n^\omega)$ time, i.e. faster than the time it takes to compute all singular values outright.\footnote{Concretely, \cite{musco2018spectrum} estimates quantities of the form $\tr(f(\A))$, with the nuclear norm equal to the case where $f(x) = |x|$. Stochastic trace estimation methods like Hutchinson's estimator and related techniques approximate $\tr(f(\A))$ by computing $\g^T f(\A)\g$ for randomly chosen $\g$ \cite{MeyerMuscoMusco:2021}.} 

Perhaps the most important example is the sum of the singular values, i.e., the Schatten 1-norm  $\|\A\|_1$ (a.k.a.~nuclear norm). \cite{musco2018spectrum} shows how to obtain a $(1+\epsilon)$ multiplicative approximation to $\|\A\|_1$ in time $\tilde O(n^{2.18}\poly(1/\epsilon))$. This is a surprising result, as no previous methods were able to beat the $O(n^\omega)$ complexity of computing a full SVD. Moreover, work in fine-grained complexity (specifically, on triangle detection lower bounds) suggests that $O(n^\omega)$ is tight for high-accuracy approximation (i.e., $\epsilon = \poly(1/n)$)\cite{musco2018spectrum,WilliamsWilliams:2010}.
By replacing the stochastic iterative methods used in \cite{musco2018spectrum} with our \AlgName{} solvers, we make further progress on this problem. 

\begin{theorem}[Schatten 1-norm estimation, corollary of Theorem \ref{t:schatten}]\label{thm:schatten}
Given an $n\times n$ matrix $\A$ and $\epsilon>0$, there is an algorithm that, with high probability, computes $X\in (1\pm\epsilon)\|\A\|_1$ in time:
\begin{align*}
    \tilde O\big(n^{2.11}\poly(1/\epsilon)\big).
\end{align*}
\end{theorem}
In fact, we show that our linear solver improves the exponent of $n$ in the time complexity of Schatten $p$-norm estimation for any $p\in(0,1.5)$, and it can also be used to approximate other matrix norms such as the Ky Fan and Orlicz norms. See Section \ref{s:schatten} for a discussion.

\subsection{Our Techniques}
\label{s:techniques}

The starting point for our \AlgFullName{} comes from work on solving regularized systems of the form $(\A+\lambda\I)\x=\b$, where $\A$ is positive definite (PD). Here, a central technique for preconditioning is the Nystr\"om method \cite{nystrom1930,Williams01Nystrom,revisiting-nystrom,musco2017recursive}, which uses randomized sketching or sub-sampling to build a low-rank approximation for $\A$. As an illustration, let $\S \in \R^{l\times n}$ be a subsampling matrix, so that $\A\S^\top$ contains a random subset of $l$ columns from $\A$. The classical Nystr\"om approximation is defined as $\hat\A_{\nys}=\C\W^{-1}\C^\top$, where $\C=\A\S^\top$ is the column submatrix and $\W=\S\A\S^\top$ is the corresponding $l\times l$ principal submatrix of $\A$. This can be extended to other sketching methods by replacing $\S$ with a sparse random matrix (CountSketch, OSNAP, LESS, etc. \cite{clarkson2013low,nelson2013osnap,cohen2016nearly,derezinski2021sparse}) or a Subsampled Randomized Hadamard Transform  (SRHT, \cite{ailon2009fast,t11}), all of which can be multiplied by $\A$ in $\tilde{O}(n^2)$ time. 

With sufficiently large sketch size (namely, proportional to the effective dimension $d_\lambda = \tr(\A(\A+\lambda\I)^{-1})$), one can show that $\M=\hat\A_{\nys}+\lambda\I$ is a good preconditioner for $\A+\lambda\I$, i.e., the condition number of $\M^{-1}(\A+\lambda\I)$ is constant. At the same time, $\M^{-1}$ can be applied quickly due to its decomposition into $\C$ and $\W$. Concretely, using the inversion formula $\M^{-1} = \frac1\lambda\big(\I-\C(\C^\top\C+\lambda\W)^{-1}\C^\top\big)$, we can compute $\M^{-1}\r$ for any vector $\r$ in $\tilde O(nl)$ time after precomputing $(\C^\top\C+\lambda\W)^{-1}$ in  $\tilde O(nl^{\omega-1})$ time.

\medskip
\noindent
\textbf{Average Condition Number Bounds via Nystr\"om Preconditioning.}
We improve on standard Nystr\"om preconditioning in two important ways. First, replacing the subsampling matrix $\S$ with a sparse sketching matrix \cite{CohenNelsonWoodruff:2016}, we leverage a moment version of the oblivious subspace embedding property in conjunction with a spectral norm error low-rank approximation bound, to show that even when solving an \emph{unregularized} system $\A\x=\b$, a preconditioner of the form $\M=\hat\A_{\nys}+\lambda\I$ can significantly improve conditioning. Note that the scalar $\lambda$ is now a parameter of the algorithm, and not of the problem, as was the case for standard Nystr\"om preconditioning. We prove that, using Nystr\"om approximation with sketch size $\tilde O(l)$ for any integer $l$, we can choose $\lambda$ so that the condition number of the preconditioned system $\M^{-1}\A$ can be reduced to $\tilde O(\frac{n}{l}\cdot\bar\kappa_l)$, where $\bar\kappa_l = \frac{1}{n-l}\sum_{i > l} \frac{\sigma_i}{\sigma_n}$ is the average condition number of $\A$, excluding its top $l$ singular values (here, $\sigma_i$ are the singular values of $\A$ listed in a decreasing order). We observe that the quality of the preconditioner exhibits a dependence on $l$ both through the $\frac nl$ factor and the condition number $\bar\kappa_l$. Thus, obtaining the best condition number requires carefully tuning the parameter $\ell$ for each application. For example, if $\A$ has at most $k$ large singular values, as in  \Cref{p:large-sv}, then $\bar\kappa_l=O(1)$ for any $l\geq k$, but due to the dependence on $\frac{n}{l}$, increasing $l$ past $k$ still improves the preconditioner. Unfortunately, this happens at the expense of increasing the $\tilde O(nl^{\omega-1})$ cost of inverting the Nystr\"om approximation, which motivates our second contribution.

\medskip
\noindent
\textbf{Two-level Preconditioning for Positive Definite Systems.}
Instead of computing $\M^{-1} = \frac{1}{\lambda}\big(\I-\C(\C^\top\C+\lambda\W)^{-1}\C^\top\big)$ \emph{explicitly}, as in standard Nystr\"om preconditioning, we focus on simply implementing matrix-vector multiplications with $\M^{-1}$, which suffice to apply a preconditioned iterative method like CG or Lanczos. The main challenge in doing so is to implement efficient multiplications with $(\C^\top\C+\lambda\W)^{-1}$, i.e., to solve the system $(\C^\top\C+\lambda\W)\x=\r$. We show how to do this using a \emph{second level of sketching and preconditioning}. In particular, since $\C$ is a tall $n \times l$ matrix, the linear system $(\C^\top\C+\lambda\W)\x=\r$ can be solved efficiently by preconditioning with $\C^\top{\mathbf{\Phi}}^\top{\mathbf{\Phi}}\C+\lambda\W$, where $\mathbf{\Phi}$ is an $O(l)\times n$ sketching matrix (specifically, an oblivious subspace embedding). The preconditioner $\C^\top{\mathbf{\Phi}}^\top{\mathbf{\Phi}}\C+\lambda\W$ can be constructed and inverted in just $\tilde{O}(nl + l^\omega)$ time, which leads to an upfront cost of $\tilde{O}(nl + l^\omega)$ for implementing $\tilde{O}(nl)$ time matrix-vector multiplications with $\M^{-1}$. This is a significant improvement on the $\tilde O(nl^{\omega-1})$ required by explicit inversion. Overall, we obtain a complexity of $\tilde O(n^2\sqrt{n/l} + l^\omega)$ for solving \Cref{p:large-sv} when $\A$ is positive definite. Optimizing over the parameter $l$ gives a final runtime bound of  $\tilde{O}(n^{2.065} + k^\omega)$ time for positive definite $\A$, which is our first step towards proving the general version of \Cref{thm:krr_formal}, which does not assume positive definiteness. See Section \ref{s:psd-proof}, Theorem \ref{thm:main_psd}, for details.

We note that to analyze the approach above, we leverage a stability analysis of the preconditioned Lanczos method, which ensures that an optimal convergence rate is still obtained even when the preconditioner $\M^{-1}$ is applied inexactly. We give this analysis in Section \ref{s:lanczos}, building on prior work on the stability of the {unpreconditioned} Lanczos method \cite{paigeThesis,paige1976error,greenbaum_89,musco_musco_sidford_18}. Many alternative approaches would also suffice. For example, we could have obtained the same running times by using preconditioned Chebyshev iteration in place of Lanczos. The robustness of Chebyshev iteration to inexact applications of $\M^{-1}$ is well understood, and leveraged, e.g., in fast solvers for Laplacian systems \cite{SpielmanTeng:2014}. However, the method tends to converge much slower in practice than more popular methods like Lanczos or the  Conjugate Gradient method. Thus, providing a stability analysis of the preconditioned Lanczos method is of independent interest beyond its application in our algorithms. Broadly, an increasing number of algorithms in numerical linear algebra combine iterative methods with inexact ``inner loops'', often applied using randomized techniques. This approach has found applications in spectral density estimation \cite{KrishnanBravermanMusco:2022}, quantum inspired linear algebra \cite{BakshiTeng:2024}, and a number of other problems \cite{orecchia2012approximating}. All of this work hinges on stability analysis akin to our results on the preconditioned Lanczos method.

\medskip
\noindent
\textbf{Three-level Preconditioning for Indefinite Systems.}
The above strategy, which involves two levels of preconditioning, is restricted to positive definite matrices. Extending the approach to arbitrary linear systems requires additional work. A natural first attempt is to reduce an arbitrary $\A\x=\b$ linear system to a positive (semi-)definite linear system via the normal equations, $\A^\top\A\x=\A^\top\b$. However, if we apply the Nystr\"om preconditioning approach discussed above, we realize that we require matrices $\C=\A^\top\A\S^\top$ and $\W=\S\A^\top\A\S^\top$, which can no longer be computed in $\tilde{O}(n^2)$ time, even if $\S$ is a sparse random matrix. To address this, we show that the matrix $(\C^\top\C+\lambda\W)^{-1}$ can nevertheless be applied efficiently by preconditioning with the approximation $(\W^2+\lambda\W)^{-1}=\frac1\lambda\big(\W^{-1}-(\W+\lambda\I)^{-1}\big)$. This preconditioner, in turn, can be applied by solving two linear systems associated with $\W^{-1}$ and $(\W+\lambda\I)^{-1}$, both using a \emph{third level} of sketching and preconditioning. Details are included in  Section~\ref{sec:proof_rec}.

\smallskip
While Theorem \ref{thm:main} focuses on solving linear systems with $k$ large singular values, our main technical result (Theorem \ref{thm:main_rec}) shows that \AlgFullName{} can be used for any linear system, with the caveat that the condition number $\bar\kappa_l$ will appear in the bounds. This more general setting is used to obtain our improved Schatten norm algorithms, which are based on the work of \cite{musco2018spectrum}, and require solving linear systems where $\bar\kappa_l$ can be bounded by $O(\sqrt{n/l})$. Optimizing over $l$ yields the final time complexity of $\tilde{O}(n^{2.11})$. See Section \ref{s:schatten} for details.

\subsection{Additional Related Work}
\label{s:related-work}
There has been significant prior work on solving linear systems with a small number of outlying singular values, or more generally, whose condition number can be reduced by eliminating large singular values.
As discussed, our work is most closely related to methods that use coarse low-rank approximations to build preconditioners for such linear systems, an approach that has been extremely popular for solving large regularized regression problems arising e.g., in Gaussian process regression (kernel regression) problems \cite{rcr17,MaBelkin:2017,MeantiCarratinoRosasco:2020,frangella2023randomized,epperly_random_precond}. 

As discussed in the previous section, techniques like Fast Randomized Hadamard Transform, sparse sketching matrices, or leverage score sampling \cite{AlaouiMahoney:2015,cohen2017input} can be used to construct a low-rank approximation in just $\tilde{O}(n^2)$ time.  When $\A$ is PSD, subquadratic time is even possible \cite{MuscoWoodruff:2017}. This is in contrast to $\ell$-rank approximation methods based on power iteration or other Krylov methods, which require $\tilde{O}(n^2\ell)$ time \cite{GonenOrabonaShalev-Shwartz:2016}. 
However, previous methods for actually applying the low-rank preconditioner required at least $O(n\ell^{\omega-1})$, motivating our \AlgFullName{}, which avoids the need for explicitly orthogonalizing an $n\times \ell$ matrix.
The goal of avoiding orthogonalization also arises in other problems related to low-rank approximation, notably the \emph{principal component regression (PCR)} problem \cite{FrostigMuscoMusco:2016,JinSidford:2019}. Our \AlgName{} methods may be able to improve PCR algorithms, which rely on black-box solvers for regularized regression.

Another related line of work seeks to understand how iterative methods like the conjugate gradient method behave for a matrix with few outlying singular values \cite{spielman2009note}. In fact, for \Cref{p:large-sv}, it is well known that the \emph{unpreconditioned} CG or Lanczos methods will converge in roughly $k + \kappa$ steps, where $\kappa$ is the condition number of $\A$'s lower $n-k$ singular values \cite{Axelsson1986}. However, the cost of such methods would still be $\tilde{O}(n^2k)$, which our work improves on. 
There has also been work on iterative methods for the conceptually related problem of solving poorly conditioned linear system consisting of multiple well-conditioned subspaces \cite{KelnerMarsdenSharan:2022}, although the challenges are different.

Also related to our work is recent progress on analyzing and improving \emph{stochastic iterative methods} like stochastic gradient descent, randomized Kaczmarz, stochastic coordinate descent, and variants thereof \cite{strohmer2009randomized,RouxSchmidtBach:2012,JohnsonZhang:2013,Shalev-ShwartzZhang:2014}. A major development in the area is that it is possible to obtain runtimes for solving linear systems that depend on the \emph{average condition} number, $\bar{\kappa} = \frac{1}{n}\sum_{i=1}^n \frac{\sigma_i}{\sigma_n}$ instead of the standard condition number $\kappa = \frac{\sigma_1}{\sigma_n}$. E.g., for PD linear systems, \cite{lee2013efficient} presents a variant of coordinate descent that runs in $\tilde{O}(n^2\sqrt{\bar{\kappa}})$ time vs. $\tilde{O}(n^2\sqrt{{\kappa}})$ for the conjugate gradient method. Similar results have been proven for variants of stochastic gradient descent \cite{FrostigGeKakade:2015}. More recently, \cite{derezinski2024fine} uses a connection between block coordinate descent and low-rank approximation \cite{derezinski2024sharp}, obtaining an algorithm that runs in $\tilde O(n^2\sqrt{\bar\kappa_\ell})$ time for $\ell = O(n^{\frac1{\omega-1}})$, where $\bar\kappa_\ell=\frac1{n-\ell}\sum_{i>\ell}\frac{\sigma_i}{\sigma_n}\leq \bar\kappa$.  Naturally, these average condition numbers $\bar{\kappa}$ and $\bar\kappa_\ell$ can be significantly smaller than $\kappa$ if $\A$ only has a few large singular values and is otherwise well-conditioned. Perhaps unsurprisingly then, the current state-of-the-art methods for \Cref{p:large-sv} and downstream applications like matrix norm approximation rely on stochastic iterative methods \cite{musco2018spectrum, derezinski2023solving,derezinski2024fine}. A high-level observation of our work is that it is possible to match and actually exceed the efficiency of these methods using deterministic iterative methods with randomized preconditioning. This will become more apparent in \Cref{sec:main_results}, where our main theorem is stated in terms of a dependence on an average condition number.

\section{Preliminaries}\label{sec:prelim}
\paragraph{Notation.}
Throughout this paper, we use $\mathcal{S}_n^+$ to denote the positive semidefinite (PSD) cone and use $\mathcal{S}_n^{++}$ to denote the positive definite (PD) cone. For vector $\x$, we use $\|\x\|$ to denote its Euclidean norm, and for a PSD matrix $\A$, we denote $\|\x\|_{\A} := \sqrt{\x^\top\A\x}$. For matrix $\A$ with singular values $\sigma_1\geq \sigma_2 \geq \ldots \geq \sigma_n$, we let $\|\A\| = \sigma_1$ and $\|\A\|_F$  denote the operator norm and Frobenius norm, respectively. We let $\kappa(\A) = \sigma_1/\sigma_n$ denote its condition number. We let $\A^\dagger = (\A^\top\A)^{-1}\A^\top$ denote the pseudoinverse of $\A$. We use $\nnz(\A)$ to denote the number of non-zero entries of $\A$. For two $n \times n$ PSD matrices $\A, \B$ and any $\epsilon>0$, we say $\A \approx_{1+\epsilon} \B$ if $\frac{1}{1+\epsilon} \A \preceq \B \preceq (1+\epsilon)\A$ holds, where $\preceq$ denotes the matrix Loewner order. In the following sections, for matrix $\A\in\R^{m\times n}$ and failure probability $\delta > 0$, we use $\tilde{O}(\cdot)$ to omit $\polylog(mn / \delta)$ factors\footnote{Notice that in comparison, in Section~\ref{s:intro} we use $\tilde{O}(\cdot)$ to hide all $\log$ factors for conciseness.}. In this paper, by saying ``with high probablity'' we mean with probability at least $1 - 1 / \poly(n)$.

For matrix $\A\in\mathcal{S}_n^+$, let $\{\lambda_i\}_{i=1}^n$ be its eigenvalues in non-increasing order. For any $\lambda > 0$, we define the effective dimension of $\A$ as $d_\lambda := \tr(\A(\A+\lambda\I)^{-1}) = \sum_{i=1}^n \frac{\lambda_i}{\lambda_i + \lambda}$. 
We will require the following basic lemma on the effective dimension:
\begin{lemma}[Lemma 5.4 in \cite{frangella2023randomized}]
\label{lem:eff_dim}
For $\A\in\mathcal{S}_n^{+}$ and $\lambda > 0$, we have the following holds:
\begin{enumerate}
    \item For any $\gamma > 0$, if $j \geq (1+\gamma^{-1}) d_\lambda$, then $\lambda_j \leq \gamma\lambda$. 
    \item If $k \geq d_\lambda$, then $\sum_{i > k} \lambda_i \leq \lambda\cdot d_\lambda$.
\end{enumerate}
\end{lemma}

\paragraph{Sparse Subspace Embedding.}
Following \cite{CohenNelsonWoodruff:2016, cohen2016nearly}, we define the sparse embedding matrix $\S\in\R^{s\times n}$ with sparsity parameter $\gamma$, which will be used in our construction of the Nystr\"om preconditioner, as follows. Notice that given $\A\in\R^{n\times n}$, $\S\A$ can be computed in time $O(\gamma\cdot \nnz(\A))$.
\begin{definition}[Sparse embedding matrix]\label{def:sparse_embed}
We define an $s \times n$ matrix $\S$ to be a sparse embedding matrix with sparsity $\gamma$, if the columns of $\S$ are independent, and for each column of $\S$, there are $\gamma$ random entries chosen uniformly without replacement and set to $\pm 1/\sqrt{\gamma}$ independently, with other entries in that column being set to $0$.
\end{definition}

Based on the construction of sparse embedding matrix, recent work \cite{chenakkod2023optimal} further constructs a fast oblivious subspace embedding (OSE) matrix and shows an optimal $O(d)$ result for embedding dimension, as stated in the following lemma. Throughout this paper, we will use  $\mathbf{\Phi}$ to denote this oblivious embedding matrix, which will be used in the construction of a preconditioner for the second level of linear system solving.
\begin{lemma}[Adapted from Theorem 1.4 in \cite{chenakkod2023optimal}]\label{lem:precondition_sparse}
Given $\A\in\R^{n \times d}, \epsilon < 1/2$, $\delta < 1/2$, in time $O(\nnz(\A)\log(d / \delta) / \epsilon + d^2\log^4(d/\delta)/\epsilon^6)$ we can compute $\mathbf{\Phi}\A$ where $\mathbf{\Phi}\in\R^{\phi \times n}$ is an embedding matrix with $\phi = O((d+\log(1/\delta)) / \epsilon^2)$, such that with probability $1-\delta$ we have
\begin{align*}
\forall \x\in\R^d, ~~\frac{1}{1+\epsilon} \|\A\x\| \leq \|\mathbf{\Phi}\A\x\| \leq (1+\epsilon) \|\A\x\|.
\end{align*}
\end{lemma}

\paragraph{Computational Model.} Our results are proven in the real RAM model (i.e, exact arithmetic). We assume the inputs $\A$ and $\b$ are given with real-valued entries, and basic arithmetic operations ($+,-,\times,\div, \sqrt{\cdot}$) can be performed exactly on real numbers in $O(1)$ time. We still need to leverage results on the finite-precision behavior of iterative solvers like the Lanczos method  \cite{paige1976error,musco_musco_sidford_18} to account for the fact that we apply preconditioners inexactly, even in the real RAM model. A full analysis of our methods in a finite precision model of computation is beyond the scope of this work, but is an important next step for future work on \AlgFullName{}.
\section{Main Technical Results}
\label{sec:main_results}

In this section, we present our main technical result in its full generality, and then discuss how it can be used to recover the claims in Section \ref{s:intro}. In this result, we address the task of solving a linear system of the form $(\A^\top\A+\lambda\I)\x=\c$ for an $m\times n$ matrix $\A$ and $\lambda\geq 0$. Note that the regularization parameter $\lambda$ is entirely optional (since we can set $\lambda=0$), and it is included here primarily for the sake of applications to kernel ridge regression and Schatten norm estimation. To recover the setting from Theorem \ref{thm:main}, i.e., solving $\A\x=\b$, we can simply rewrite this problem via the normal equations as $\A^\top\A\x=\A^\top\b$, then set $\c=\A^\top\b$ and $\lambda=0$.

\begin{theorem}[Main technical result]\label{thm:main_rec}
Given $\A\in\R^{m\times n}$ with condition number $\kappa$, vector $\b\in\R^n$ and regularization parameter $\lambda\geq 0$, let $\{\sigma_i\}_{i=1}^n$ be the singular values of $\A$ in decreasing order, and let $\x^* = (\A^\top\A + \lambda\I)^{-1}\c$.
For any $l \in \{\log n + 1,\ldots,  n\}$, define $\bar{\kappa}_{l,\lambda} \coloneqq (\frac{1}{n-l}\sum_{i > l} \sigma_i^2 / (\sigma_n^2+\lambda))^{1/2}$. Given $\epsilon > 0$ and $\delta \in (0,1/8)$,  if Algorithm~\ref{alg:msp_rec_main} is run with $\lambda_0 = \frac{2}{l}\sum_{i> l}\sigma_i^2, s = O(l\log(l/\delta)), \gamma = O(\log(l/\delta)), \phi = O(s + \log 1/\delta)$ and $t_{\max} = O(\sqrt{n/l} \cdot \bar{\kappa}_{l,\lambda}\log (\bar{\kappa}_{l,\lambda}/\epsilon))$, then with probability at least $1-\delta$, it will output $\tilde{\x}$ such that $\|\tilde{\x} -\x^*\|_{\A^\top\A+\lambda\I}\leq\epsilon\|\x^*\|_{\A^\top\A+\lambda\I}$ in time
\begin{align*}
\tilde{O}\left(\nnz(\A)\sqrt{\frac{n}{l}} \bar{\kappa}_{l,\lambda}\log^3 (\kappa/\epsilon) + l^{\omega}\right)
\end{align*}
where $\tilde{O}(\cdot)$ hides $\polylog(mn/\delta)$ factors. Moreover, if instead of matrix $\A$, we are given the matrix $\A^\top\A$ directly, then the same time complexity can be achieved with $\nnz(\A)$ replaced by $\nnz(\A^\top\A)$ by using Algorithm~\ref{alg:msp_psd}. See Theorem~\ref{thm:main_psd} for details.
\end{theorem}
\begin{remark}
    The quantity $\bar\kappa_{l,\lambda}$ above has appeared (in similar forms) in prior work on the convergence of stochastic iterative methods for linear systems \cite{lee2013efficient,musco2018spectrum,derezinski2023solving}. When $\lambda=0$, this quantity satisfies $1\leq \bar\kappa_{l,0}\leq \sigma_{l+1}/\sigma_n$, and it behaves like a typical averaged condition number of $\A$. However, when $\lambda>0$, $\bar\kappa_{l,\lambda}$ can in fact be much less than $1$, which we exploit in the application to kernel ridge regression.
\end{remark}
We next demonstrate how Theorem \ref{thm:main_rec} can be used to show our main result for matrices with $k$ large singular values. To do so, we will optimize over the choice of $l$ as follows.

\begin{proof}[Proof of Theorem~\ref{thm:main}]
As mentioned above, to solve the consistent linear system $\A\x=\b$, we can apply Theorem~\ref{thm:main_rec} with choice $\lambda=0$ and $\c = \A^\top\b$. The solution $\x^*$ of the resulting linear system $\A^\top\A\x = \A^\top\b$, $(\A^\top\A)^{-1}\A^\top\b = \A^\dagger \b$, is the minimum norm solution of the general linear system $\A\x = \b$. Additionally, the output guarantee from Theorem~\ref{thm:main_rec}, that $\|\tilde{\x} - \x^*\|_{\A^\top \A} \leq \epsilon \|\x^*\|_{\A^\top \A}$, ensures that $\|\A\tilde{\x} - \b\| \leq \epsilon\|\b\|$ as required by Theorem~\ref{thm:main}. 

It remains to obtain the desired time complexity by optimizing over $l$, which we will choose  to be larger than  $k$.
Notice that for $l > k$, we have $\bar{\kappa}_{l,0} \leq \sigma_{l+1}/\sigma_n = O(1)$ under the assumption of Theorem~\ref{thm:main}. So, upper bounding $\nnz(\A) \leq n^2$, we can upper bound the runtime in Theorem \ref{thm:main_rec} by (here $\tilde{O}(\cdot)$ hide the logarithmic dependencies on $n$ and $\kappa$):
\begin{align*}
    \tilde O(n^{2.5}/\sqrt l \cdot \log^3 1 / \epsilon + l^\omega),\quad\text{for any $l\geq k$}.
\end{align*}
Optimizing over $l$ so that both terms are proportionally large, we obtain the following two cases:
\begin{itemize}
    \item If $k \leq n^{5/(2\omega+1)}$, we choose $l = n^{5/(2\omega+1)}$ and obtain time complexity $\tilde{O}(n^{2+\frac{\omega-2}{2\omega+1}} \log^3 1/\epsilon)$.
    \item If $n^{5/(2\omega+1)} < k < n$, we choose $l = k$, and obtain time complexity $\tilde O (n^{2+\frac{\omega-2}{2\omega+1}} \log^3 1/\epsilon+k^{\omega})$.
\end{itemize}
Thus, the overall time complexity becomes:
\begin{align*}
    \tilde O(n^{2+\frac{\omega-2}{2\omega+1}} \log^3 1/\epsilon + k^\omega),
\end{align*}
which simplifies to $\tilde O(n^{2.065}\log^3 1/\epsilon+k^\omega)$ with current matrix multiplication exponent $\omega \approx 2.372$.
\end{proof}

\subsection{Applications to Regularized Linear Systems and Least Squares}
Some of the most important applications of linear systems, particularly in the context of machine learning and statistics, are regularized and unregularized least squares problems. In this section, we show how Theorem \ref{thm:main_rec} can be adapted and applied to these settings.

\paragraph{Kernel ridge regression.}
One of the most computationally expensive variants of these tasks arises when we consider kernel-based learning methods \cite{rcr17,MaBelkin:2017,MeantiCarratinoRosasco:2020}. These approaches work by implicitly constructing expanded high-dimensional representations of data points, which are accessed only through inner products, which are computed using a kernel function. For a dataset with $n$ datapoints, this gives rise to an $n\times n$ positive definite kernel matrix $\K$, whose $(i,j)$ entry is the inner product between the $i$th and $j$th data point in the expanded feature representation. The resulting prediction model forms the kernel ridge regression (KRR) task: 
\begin{align*}
    \min_{\x}\ \frac1n\sum_{i=1}^n([\K\x]_i-y_i)^2 + \frac\lambda 2\x^\top\K\x,
\end{align*}
where $[\K\x]_i$ is the $i$th entry of the length $n$ vector $\K\x$ and $y_1, \ldots, y_n$ are training labels. KRR is equivalent to solving the regularized positive definite linear system $(\K + n\lambda\I)\x = \y$. Since $\K$ is $n\times n$, the cost of solving this problem exactly scales with $O(n^\omega)$. Regularized linear systems of this form also arise independently in continuous second-order optimization methods, where instead of the kernel matrix $\K$, we consider a Hessian matrix $\H$, and the regularization term $\lambda$ is a parameter of the optimization algorithm \cite{levenberg1944method,marquardt1963algorithm,monteiro2012iteration,jiang2023online}. 

We show how to adapt our main result to the setting above, giving a faster method for solving any linear system of the form $(\A+\lambda\I)\x=\b$ for a positive semidefinite matrix $\A$. Our method improves on the previously best known time complexity of this problem \cite{acw17,frangella2023randomized}, from $\tilde O(n^2+n{d_\lambda}^{\omega-1})$ to $\tilde O(n^2+{d_\lambda}^\omega)$, where recall that $d_\lambda$ is the $\lambda$-effective dimension of $\A$. We note that this result takes advantage of the fact that $\bar\kappa_{l,\lambda}$ can in fact be smaller than $1$ for $\lambda>0$.

\begin{theorem}[Regularized linear systems, formal version of Theorem~\ref{thm:krr}]\label{thm:krr_formal}
Consider positive semidefinite $\A\in\mathcal{S}_n^+$ with condition number $\kappa$ and effective $d_{\lambda} = \tr(\A(\A+\lambda\I)^{-1})$ for $\lambda > 0$. For a target $\b\in\R^n$, let $\x^* = (\A+\lambda\I)^{-1}\b$. Given $\epsilon>0$ and $0<\delta<1/8$, with probability $1-\delta$ we can compute $\tilde{\x}$ such that $\|\tilde{\x} - \x^*\|_{\A+\lambda\I} \leq\epsilon\|\x^*\|_{\A+\lambda\I}$ in time
\begin{align*}
\tilde{O}\left(n^2\cdot \log^3 (\kappa/\epsilon) + {d_{\lambda}}^{\omega}\right).
\end{align*}
\end{theorem}

\begin{proof}
Without loss of generality, we assume that $d_{\lambda} < n/4$. If $d_\lambda$ is larger, the stated runtime follows by simply using a direct $O(n^\omega)$ time method for inverting $(\A + \lambda \I)$.

We obtain the result by applying Theorem~\ref{thm:main_rec} to the matrix $\A^{1/2}$. As stated in the last lines of the theorem, it holds even if we only have access to $(\A^{1/2})^T\A^{1/2} = \A$. We show that if we choose $l = 2d_{\lambda}$, we have $\bar{\kappa}_{l,\lambda}^2(\A^{1/2}) = O(l/n)$. In particular, recall that $d_{\lambda} = \tr(\A(\A+\lambda\I)^{-1})$ and, from Lemma~\ref{lem:eff_dim} we know that if we choose $l = 2d_{\lambda}$ then $\lambda_l(\A) \leq \lambda$. We thus have: 
\begin{align*}
l = 2d_{\lambda} = 2\sum_{i=1}^n \frac{\lambda_i(\A)}{\lambda_i(\A) + \lambda} > 2\sum_{i=l}^n \frac{\lambda_i(\A)}{\lambda_i(\A) + \lambda} \geq \frac{2\cdot \sum_{i>l} \lambda_i(\A)}{\lambda_l(\A) + \lambda} \geq \frac{1}{\lambda} \sum_{i>l} \lambda_i(\A).
\end{align*}
By applying this result we have
\begin{align*}
\bar{\kappa}_{l,\lambda}^2(\A^{1/2}) = \frac{1}{n-l}\sum_{i>l}\frac{\lambda_i(\A)}{\lambda_n(\A) + \lambda} \leq \frac{1}{n-l}\cdot\frac{l\lambda}{\lambda_n(\A)+\lambda}\leq \frac{l}{n-l} \leq \frac{2l}{n}.
\end{align*}
In the last step we use that, since $d_{\lambda} \leq n/4$ by assumption, $l \leq n/2$.

Now we simply plug into Theorem~\ref{thm:main_rec}, noting that $\sqrt{\frac{n}{l}}\bar{\kappa}_{l,\lambda}$ can be upper bounded by $\sqrt{2}$ given the result above. 
We conclude that we can compute $\tilde{\x}$ such that $\|\tilde{\x}-\x^*\|_{\A+\lambda\I} \leq \epsilon \|\x^*\|_{\A+\lambda\I}$ in time $\tilde{O}(n^2 \cdot \log^3 (\kappa/\epsilon) + {d_{\lambda}}^{\omega}).$
\end{proof}

\paragraph{Least Squares.} Our Theorem~\ref{thm:main_rec} naturally extends to over-determined linear systems. For simplicity we state a result in the setting of Problem \ref{p:large-sv}, where $\A$ has at most $k$ large singular values and is otherwise well conditioned.

\begin{corollary}[Least squares]\label{cor:least_squares}
Given matrix $\A\in\R^{n \times d}$ with at most $k$ singular values larger than $O(1)$ times its smallest singular value, $\b \in \R^n, \epsilon > 0$ and $0<\delta<1/8$, with probability at least $1-\delta$, we can compute $\tilde \x$ such that $\|\A\tilde\x-\b\|^2\leq\min_\x\|\A\x-\b\|^2 + \epsilon\|\b\|^2$ in time
\begin{align*}
\tilde{O}\left((\nnz(\A) + d^{2.065})\cdot\log 1/\epsilon + k^\omega\right).
\end{align*}
\end{corollary}
This result is obtained by combining our MSP method with a standard sketch-and-precondition iterative solver \cite{rokhlin2008fast,epperly23}. In particular, for a tall matrix $\A\in\R^{n\times d}$ where $n \gg d$, one can first construct a sketch $\tilde{\A}$ with $O(d)$ rows using a constant-factor sparse oblivious subspace embedding (e.g., Lemma~\ref{lem:precondition_sparse}) which takes $\tilde{O}(\nnz(\A) + d^2)$ time. 
To solve a least squares problem involving $\A$, it suffices to solve $O(\log(1/\epsilon))$ linear systems of the form $(\tilde{\A}^\top\tilde{\A})^{-1}\g$ for some vector $\g$ up to constant error. See e.g. Section 10 of \cite{derezinski2023solving} for details. 

To solve these required linear systems, we apply \Cref{thm:main_rec} directly. We note that, since $\tilde{\A}$ is a subspace embedding of $\A$, all of its singular values are within a multiplicative constant factor of those of $\A$. Accordingly, it too has at most $k$ singular values greater than a constant times its smallest singular value. 
We conclude that, for $l\geq k$, each application of \Cref{thm:main_rec} takes $\tilde{O}(\nnz(\A) + d^2\sqrt{d/l})$ time. The additive $O(l^{\omega})$ term in the runtime of \Cref{thm:main_rec} comes from the construction of an inner preconditioner for $\tilde{\A}$, so is only incurred a single time. 
Optimizing over the choice of $l$ yields Corollarly \ref{cor:least_squares}. Notice that we only incur a $\log 1/\epsilon$ dependence in the result instead of $\log^3 1/\epsilon$. This is because each system involving $\tilde{\A}$ only needs to be solved to constant accuracy for the overall sketch-and-precondition method to converge in $O(\log 1/\epsilon)$ iterations.

\subsection{Applications to Matrix Norm Estimation}
\label{s:schatten}

We next discuss how our methods can be used in the task of estimating the Schatten norm of a matrix, providing a direct improvement over a result obtained in \cite{musco2018spectrum}. 
For the sake of simplicity, we focus on the case of a dense square matrix $\A\in\R^{n\times n}$. However, using similar strategies to those described in \cite{musco2018spectrum}, our results can also be used to improve time complexities when $\A$ is sparse or rectangular.

To obtain our results, we note that the norm estimation algorithms of \cite{musco2018spectrum} are in fact meta-algorithms that require a black-box linear solver for regularized systems. Therefore any improvement to the time complexity of the black-box solver leads to a potential improvement in the complexity of Schatten norm estimation (as well as to other spectrum approximation tasks discussed in that paper). The key black-box requirement from \cite{musco2018spectrum} is as follows: Given a matrix $\A\in\R^{n\times n}$, parameter $\lambda>0$, tolerance $\epsilon\in(0,1)$, and a vector $\y\in\R^n$, return a vector $\x$ such that:
\begin{align}
    \|\x - \M_\lambda^{-1}\y\|_{\M_\lambda}\leq \epsilon\|\y\|_{\M_\lambda^{-1}},\quad\text{where}\quad \M_\lambda=\A^\top\A+\lambda\I.\label{eq:schatten-ridge}
\end{align}
A key parameter in the analysis of \cite{musco2018spectrum} for solving such linear systems turns out to be the same  averaged tail condition number that arises in our analysis (up to adjustments in notation):
\begin{align*}
\bar\kappa_{k,\lambda}(\A) \coloneqq \Big(\frac1{n-k}\sum_{i>k}\frac{\sigma_i^2(\A)}{\sigma_{\min}^2(\A)+\lambda}\Big)^{1/2}.
\end{align*}
In particular, the main result of \cite{musco2018spectrum} for Schatten norm estimation can be reformulated as follows:
\begin{lemma}[adapted from the proof of  Corollary 12 in \cite{musco2018spectrum}]\label{l:schatten-blackbox}
    For any matrix $\A\in\R^{n\times n}$, parameter $p\in(0,2)$,
    $\epsilon\in(0,1)$, and a parameter $k\leq n$, there is an algorithm that, using $\tilde
    O(\frac1{\epsilon^5p})$ calls to an $n\times n$ ridge regression
problem \eqref{eq:schatten-ridge} with $\bar\kappa_{k,\lambda} \leq \frac1{\epsilon^{1/p}}(\frac nk)^{1/p-1/2}$, returns
    $X\in (1\pm\epsilon)\|\A\|_p^p$.
\end{lemma}
We note that a similar statement applies to the general spectral sums approximation result of \cite{musco2018spectrum} (Theorem 11), which can also by used to approximate other norms like the Ky Fan and Orlicz norm. We focus on the Schatten norms merely for the sake of simplicity and conciseness. Also, we only focus on Schatten norms with $0<p<2$, including the most important case of $p=1$, as the existing algorithms for the case of $p\geq 2$ are already near-optimal.

We are now ready to state our result for Schatten norm estimation, by combining Lemma \ref{l:schatten-blackbox} with our linear system solver.
\begin{theorem}\label{t:schatten}
    For any $p\in(0,2)$ and $\A\in\R^{n\times n}$ there is an algorithm that, with high probability, returns $X\in(1\pm\epsilon)\|\A\|_p^p$ which runs in time:
    \begin{align*}
        \tilde O\Big(n^{2+\frac{\omega-2}{p\omega+1}}\poly(1/\epsilon)\Big).
    \end{align*}
    In particular, for $p=1$, we can get a constant factor approximation to the Schatten 1-norm of $\A$ in $\tilde O(n^{2.11})$ time with fast matrix multiplication and in $\tilde O(n^{2.25})$ time without fast matrix multiplication. 
\end{theorem}
\begin{remark}
    This result improves the exponent of $n$ in the time complexity, compared to \cite{musco2018spectrum}, for any $0<p<1.5$. In particular, for $p=1$, we improve from $\tilde O(n^{2.18})$ to $\tilde O(n^{2.11})$ with fast matrix multiplication, and from $\tilde O(n^{2.33})$ to $\tilde O(n^{2.25})$ without fast matrix multiplication. While we did not optimize the time complexity dependence on $\epsilon$, it should be possible to recover the $O(\frac1{\epsilon^{\max\{3,1+1/p\}}})$ dependence from \cite{musco2018spectrum} by carefully repeating their optimized analysis given in Theorem 31.
\end{remark}
\begin{proof}
We will prove this result by applying our main result, Theorem \ref{thm:main_rec}. However, observe that this theorem has a logarithmic dependence on the condition number of $\A$, which we would like to avoid. To do so, observe that we can reduce to the case where the input $\A$ has condition number $\poly(n/\epsilon)$. If it does not, we can simply add $z \I$ to $\A$, where $\I$ is an $n\times n$ identity matrix and $z$ is a mean zero Gaussian random variable with standard deviation $\|\A\|\cdot \poly(\epsilon/n)$. $\|\A\|$ can be approximated to multiplicative accuracy in $\tilde{O}(n^2)$ time using standard power method. With high probability, we will have that $|z| \leq \|\A\|\cdot \poly(\epsilon/n)$, so the Schatten $p$-norm of $\A + z\I$ is within a multiplicative $(1\pm \epsilon/2)$ factor of that of $\A$. Moreover, by a union bound and standard anti-concentration of Gaussian random variables, we will have that all eigenvalues of $\A + z\I$ are at least $\|\A\|\cdot \poly(\epsilon/n)$ far away from zero with high probability, leading to condition number $\poly(n/\epsilon)$.

With the condition number bounded, for any $\lambda>0$ and $k$, using our main linear system solver from Theorem \ref{thm:main_rec}, we can solve the regularized linear system \eqref{eq:schatten-ridge} in time $\tilde O(\nnz(\A)\sqrt{\frac nk}\bar\kappa_{k,\lambda} + k^\omega)$. Focusing on dense matrices so that $\nnz(\A)=n^2$, and assuming that $\bar\kappa_{k,\lambda} \leq \frac1{\epsilon^{1/p}}(\frac nk)^{1/p-1/2}$ as in Lemma \ref{l:schatten-blackbox}, we obtain the following overall time complexity for Schatten norm estimation:
\begin{align*}
\tilde O\Bigg(\frac1{p\epsilon^5}n^2\sqrt{\frac nk}\epsilon^{-1/p}\Big(\frac nk\Big)^{1/p-1/2} + k^\omega\bigg) = \tilde O\Big(\frac1{p\epsilon^{5+1/p}}\frac{n^{2+1/p}}{k^{1/p}} + k^\omega\Big),
\end{align*}
where we use the fact that the $\tilde O(k^\omega)$ (needed to build our recursive preconditioner) only has to be performed once and can be reused for all linear solves. Now, it remains to optimize the time complexity over $k$. Without optimizing over $\epsilon$ dependence, we simply balance out the two terms $\frac{n^{2+1/p}}{k^{1/p}}$ and $k^\omega$, obtaining $n^{\frac{2p\omega+\omega}{p\omega+1}} = n^{2+\frac{\omega-2}{p\omega+1}}$.
\end{proof}

\section{Two-Level \AlgName{} for Positive Definite Linear Systems}\label{s:psd-proof}
Before proceeding into the proof of our main result Theorem~\ref{thm:main_rec}, we first consider a special (and easier) case of solving $(\A+\lambda\I)\x = \b$, where $\A\in\mathcal{S}_n^{++}$ and $\lambda \geq 0$. 
Let $\{\lambda_i\}_{i=1}^n$ be the eigenvalues of $\A$ in decreasing order. Given $l<n$ and $\lambda\geq 0$, we define the average regularized tail condition number of $\A$ as $\bar{\kappa}_{l,\lambda}(\A) := \frac{1}{n-l}\sum_{i>l}^n\frac{\lambda_i}{\lambda_n+\lambda}$. Notice that $\bar{\kappa}_{l,\lambda}(\A)$ is decreasing in both $l$ and $\lambda$. For simplicity, we denote it as $\bar{\kappa}_{l}$ if we assume that $\lambda=0$ and matrix $\A$ is clear from the context.

We solve this positive definite linear system with preconditioned Lanczos method. The following Lemma~\ref{lem:stable_lanczos} guarantees that, when using a preconditioner $\M$, if for given vector $\r$ we can approximate $\M^{-1}\r$ well enough and satisfy Eq.\eqref{eq:lanczos_assump}, then preconditioned Lanczos can converge in roughly the same number of iterations as if we compute $\M^{-1}\r$ exactly. Details are in Section \ref{s:lanczos}.

\begin{lemma}[Preconditioned Lanczos, restated Theorem~\ref{thm:stable_lanczos}]\label{lem:stable_lanczos}
Consider solving $\A\x = \b$ for positive definite $\A$ using preconditioned Lanczos provided with a function $\mathsf{SolveM}$ that, for some preconditioner $\M$ and any vector $\r$ returns 
\begin{align}\label{eq:lanczos_assump}
\|\mathsf{SolveM}(\r) - \M^{-1}\r\|_{\M} \leq \epsilon_0 \cdot \|\M^{-1}\r\|_{\M}.
\end{align}
If $\epsilon_0 \leq (\frac{\epsilon}{\kappa_{\M} n})^c$ for a fixed constant $c>0$, where $\kappa_{\M}$ is the condition number of $\M^{-1/2}\A\M^{-1/2}$,
Then, preconditioned Lanczos (Algorithm \ref{alg:lanczos_left_precon}) with $t = O(\sqrt{\kappa_{\M}}\log(\kappa_{\M}/\epsilon))$ iterations returns $\tilde{\x}$ s.t.:
\begin{align*}
\|\tilde{\x} - \x^*\|_{\A} \leq \epsilon \cdot \|\x^*\|_{\A} ~~~\text{where}~~~ \x^* = \A^{-1}\b.
\end{align*}
\end{lemma}
The construction of preconditioner $\M$ is based on Nystr\"om approximation. Let $\C = \A\S^\top$ and $\W = \S\A\S^\top$ where $\S \in \R^{s \times n}$ is some sketching matrix. Then the Nystr\"om approximation of $\A$ can be expressed as $\hat{\A}_{\nys} = \C\W^{-1}\C^\top$. We show that if we choose $\S$ to be a sparse embedding matrix (according to Definition~\ref{def:sparse_embed}) with $s = \tilde{O}(l)$ rows and $\gamma = \tilde{O}(1)$ non-zero entries per column, then we have the approximation guarantee $\|\hat{\A}_{\nys} - \A\| \leq \frac{2}{l} \sum_{i>l} \lambda_i$, which in turn gives $\kappa_{\M} = O(\bar{\kappa}_{l,\lambda} n / l)$ for $\M \coloneqq \hat{\A}_{\nys} + \tilde{\lambda}\I$ where $\tilde{\lambda} \coloneqq \lambda + \frac{2}{l} \sum_{i>l} \lambda_i$, see Lemma~\ref{lem:precondition_sparse_main}.
With the choice of $\M$ we can verify the following holds:
\begin{align*}
\M^{-1} \r = \frac{1}{\tilde{\lambda}} \left(\r - \C(\C^\top\C+\tilde{\lambda}\W)^{-1}\C^\top\r\right).
\end{align*}
To approximate $\M^{-1}\r$ we can first solve the linear system $(\C^\top\C + \tilde{\lambda}\W)\y = \C^\top\r$ and get approximate solution $\hat{\y}$, then use $\hat{\w} \coloneqq \frac{1}{\tilde{\lambda}}(\r - \C\hat{\y})$ as the approximator, see Lemma~\ref{lem:trans_variable}. However, constructing such a linear system takes $O(ns^2)$ which is unaffordable for us. Instead, we introduce a second level, and again use preconditioned Lanczos as the Level-2 solver to solve the linear system $(\C^\top\C+\tilde{\lambda}\W)\y = \C^\top\r$ without explicitly computing $\C^\top\C$. Compared with the first level where we do the preconditioning coarsely, in the second level we obtain an optimal $O(1)$ condition number after preconditioning.

Let $\mathbf{\Phi} \in \R^{\phi \times n}$ be the oblivious subspace embedding matrix from Lemma~\ref{lem:precondition_sparse} with $\phi = \tilde{O}(s)$. We can construct the Level-2 preconditioner as $\M_2 \coloneqq \tilde{\C}^\top\tilde{\C}+\lambda\W$ where $\tilde{\C} = \mathbf{\Phi}\C$ in time $\tilde{O}(ns+s^\omega)$. Notice that this step only needs to be done once. By transforming the error bound (that we get from Level-2 preconditioned Lanczos) from $\hat{\y}$ to $\hat{\w}$ using Lemma~\ref{lem:trans_variable}, we can verify the assumption Eq.\eqref{eq:lanczos_assump} which is required in Level-1 preconditioned Lanczos, thus finishing the proof. Formally, we have the following theorem for solving positive definite linear systems, where note that the condition number $\bar\kappa_{l,\lambda}$ is defined differently in terms of the eigenvalues of $\A$, rather than in terms of the squared singular values. Also note that although we assume $\A\in\mathcal{S}_n^{++}$, indeed we only need $\A\in\mathcal{S}_n^+$ and $\A+\lambda\I\in\mathcal{S}_n^{++}$ (which is satisfied if $\lambda>0$).

\begin{theorem}[\AlgName{}, positive definite]\label{thm:main_psd}
Given $\A\in\mathcal{S}_n^{++}$ with condition number $\kappa$, $\b\in\R^n$ and regularized term $\lambda\geq 0$. Let $\{\lambda_i\}_{i=1}^n$ be the eigenvalues of $\A$ in decreasing order, let $\x^* = (\A+\lambda\I)^{-1}\b$ be the solution of the regularized linear system $(\A+\lambda\I)\x = \b$. For any $\log n < l < n$, we define $\bar{\kappa}_{l,\lambda} \coloneqq \frac{1}{n-l}\sum_{i > l}^n \lambda_i / (\lambda_n+\lambda)$.
Given $\epsilon > 0$ and $\delta<1/8$, with probability at least $1-\delta$, running Algorithm~\ref{alg:msp_psd} with choice $\lambda_0 = \frac{2}{l}\cdot\sum_{i>l} \lambda_i, s = O(l \log(l/\delta)), \gamma=O(\log(l/\delta)), \phi = O(s + \log 1/\delta)$ and $t_{\max} = O(\sqrt{
\bar{\kappa}_{l,\lambda} n/l} \cdot\log (\bar{\kappa}_{l,\lambda}n/\epsilon))$ will output $\tilde{\x}$ such that $\|\tilde{\x} - \x^*\|_{\A+\lambda\I} \leq \epsilon \|\x^*\|_{\A+\lambda\I}$ in time
\begin{align*}
\tilde{O}\left(\nnz(\A)\sqrt{\frac{\bar{\kappa}_{l,\lambda} n}{l}}\log^2(\kappa/\epsilon) + l^{\omega}\right)
\end{align*}
where $\tilde{O}(\cdot)$ hides $\polylog(n/\delta)$ factors.
\end{theorem}

\begin{algorithm}[!ht]
\caption{\AlgName{} for solving regularized positive definite linear system $(\A+\lambda\I)\x = \b$.}
\label{alg:msp_psd}
\begin{algorithmic}[1]
\State \textbf{Input: }matrix $\A\in\mathcal{S}_n^{++}$, regularizer terms $\lambda$ and $\lambda_0$, vector $\b \in \R^n$, sparse sketch size $s$, \# of non-zeros $\gamma$, level-2 sketch size $\phi$, \# of iterations $t_{\max}$;
\State Construct sparse embedding matrix $\S\in\R^{s\times n}$ with $\gamma$ non-zeros per column; \Comment{Def.~\ref{def:sparse_embed}.}
\State Compute $\tilde{\lambda} = \lambda+\lambda_0$, $\C = \A\S^\top$, $\W = \S\C$;
\State Compute $\tilde{\C} = \mathbf{\Phi}\C \in \R^{\phi \times s}$ and $\M_2^{-1} = (\tilde{\C}^\top\tilde{\C} + \tilde{\lambda}\W)^{-1}$; \Comment{Lemma~\ref{lem:precondition_sparse}.}
\State Solve $\tilde{\x}$ by calling preconditioned Lanczos with $(\A+\lambda\I, \b, \mathsf{SolveM1}, t_{\max})$; \Comment{Alg.~\ref{alg:lanczos_left_precon}.} \\
\Return $\tilde{\x}$; \Comment{Solves $(\A+\lambda\I)\x = \b$.}
\end{algorithmic}
\end{algorithm}

\begin{algorithm}[!ht]
\caption{Level-1 auxiliary function $\mathsf{SolveM1}$ for solving $\M\w=\r$ (positive definite case).}
\label{alg:msp_psd_func}
\begin{algorithmic}[1]
\State \textbf{function} $\mathsf{SolveM1}(\r)$:
\State \quad Solve $\hat{\y}$ by preconditioned Lanczos with $(\C^\top\C+\tilde{\lambda}\W, \C^\top\r, \M_2^{-1}, O(\log \kappa/\epsilon_0))$; \Comment{Alg.~\ref{alg:lanczos_left_precon}.}
\State \quad Compute $\hat{\w} = \frac{1}{\tilde{\lambda}}(\r - \C \hat{\y})$;
\\
\Return $\hat{\w}$; \Comment{Solves $\M\w = \r$ for $\M = \C\W^{-1}\C^\top + \tilde{\lambda}\I$.}
\end{algorithmic}
\end{algorithm}

\subsection{Proof of Theorem~\ref{thm:main_psd}}

We give the proof in the following four parts.

\paragraph{Part 1: Nystr\"om Preconditioner with Sparse Embedding.}
Let $\S\in\R^{s\times n}$ be the sparse embedding matrix according to Definition~\ref{def:sparse_embed} with $s = O(l\log(l/\delta))$ and $\gamma = O(\log(l/\delta))$. Denote $\C\coloneqq \A\S^\top$ and $\W \coloneqq \S\A\S^\top$, then we have $\hat{\A}_{\nys} = \C\W^{-1}\C^\top$ and $\M = \hat{\A}_{\nys}+\tilde{\lambda}\I$ where $\tilde{\lambda} = \lambda+\lambda_0$. According to Lemma~\ref{lem:precondition_sparse_main}, with probability at least $1-\delta/2$ we have $\kappa_{\M} \leq \frac{C\bar{\kappa}_{l,\lambda} n}{l}$ for some $C=O(1)$. For the proof of Lemma~\ref{lem:precondition_sparse_main} and a detailed discussion of sparse embedding matrices, see Section~\ref{sec:sparse_embed}.

\begin{lemma}[Preconditioner based on sparse embedding]\label{lem:precondition_sparse_main}
Given positive definite matrix $\A\in\R^{n\times n}$ and $\delta \in (0,1/2)$. For $\log n < l < n$, let $\hat{\A}_{\nys} = \A\S^\top(\S\A\S^\top)^{-1}\S\A$ be the Nystr\"om approximation of $\A$, where $\S\in\R^{s\times n}$ is the sparse embedding matrix with $s = O(l\log(l/\delta))$ and $\gamma = O(\log(l/\delta))$. Given $\lambda\geq 0$, let $\M = \hat{\A}_{\nys} + (\lambda+\lambda_0)\I$ where $\lambda_0 = \frac{2}{l}\sum_{i\geq l+1}\lambda_i$, then with probability $1-\delta$ the condition number $\kappa_{\M}$ of the matrix $\M^{-1/2}\A\M^{-1/2}$ satisfies:
\begin{align*}
\kappa_{\M} \leq \frac{C\bar{\kappa}_{l,\lambda} n}{l} ~~~\text{for some}~~~ C = O(1).
\end{align*}
\end{lemma}

\paragraph{Part 2: Level-1 Preconditioning.}
For the first level we use preconditioned Lanczos (Algorithm~\ref{alg:lanczos_left_precon}) to solve $(\A+\lambda\I)\x= \b$ with preconditioner $\M$, where we need to compute $\M^{-1} \r$ in each iteration. The following lemma guarantees that by using inverse formula, we can approximate $\M^{-1}\r$ well enough by solving a linear system with matrix $\C^\top\C+\tilde{\lambda}\W$ to a certain accuracy. For the proof of Lemma~\ref{lem:trans_variable}, see Section~\ref{sec:trans_variable}.
\begin{lemma}\label{lem:trans_variable}
Given matrix $\A\in\mathcal{S}_n^{++}$, let $\hat{\A}_{\nys} = \C\W^{-1}\C^\top$ be its Nystr\"om approximation where $\C = \A\S^\top$ and $\W = \S\A\S^\top$. For $\tilde{\lambda}>0$, denote $\M = \hat{\A}_{\nys} + \tilde{\lambda}\I$ as the Nystr\"om preconditioner and assume that $\M \approx_2 \A+\tilde{\lambda}\I$. Given vector $\r\in\R^n$ and $\epsilon_1>0$, suppose we can solve the linear system $(\C^\top\C+\tilde{\lambda}\W)\y = \C^\top\r$ and compute $\hat{\y}$ such that
\begin{align*}
\|\hat{\y} - \y^*\|_{\C^\top\C+\tilde{\lambda}\W} \leq \epsilon_1 \cdot \|\y^*\|_{\C^\top\C+\tilde{\lambda}\W} ~~~\text{where}~~~\y^* \coloneqq (\C^\top\C+\tilde{\lambda}\W)^{-1}\C^\top\r.
\end{align*}
Then, $\hat{\w} = \frac{1}{\tilde{\lambda}}(\r-\C\hat{\y})$ is an approximate solution of the linear system $\M\w = \r$ satisfying
\begin{align}\label{eq:y_to_z}
\|\hat{\w} - \M^{-1}\r\|_{\M} \leq c_0\epsilon_1\kappa^{3/2} \|\M^{-1}\r\|_{\M} ~~~\text{for some}~~~c_0=O(1).
\end{align}
\end{lemma}
\noindent
According to Eq.\eqref{eq:M_inverse} in the proof of Lemma~\ref{lem:trans_variable}, we have the following holds:
\begin{align*}
\M^{-1} = \frac{1}{\tilde{\lambda}} \left(\I - \C(\C^\top\C+\tilde{\lambda}\W)^{-1}\C^\top\right).
\end{align*}
By using this inverse formula, we have $\w^* \coloneqq \M^{-1}\r = \frac{1}{\tilde{\lambda}}(\r - \C\y^*)$. Instead of directly computing $\y^*$, suppose we can solve the linear system $(\C^\top\C+\tilde{\lambda}\W)\y = \C^\top\r$ approximately and have the guarantee (which will be proved in Level-2 analysis later) that $\|\hat{\y} - \y^*\|_{\C^\top\C+\tilde{\lambda}\W} \leq \epsilon_1\cdot \|\y^*\|_{\C^\top\C+\tilde{\lambda}\W}$, then Lemma~\ref{lem:trans_variable} guarantees that Eq.\eqref{eq:y_to_z} holds as long as $\M \approx_2 \A + \tilde{\lambda}\I$. The assumption here can be easily verified: for one side we have
\begin{align*}
\M = \hat{\A}_{\nys} + \tilde{\lambda}\I \preceq \A + \tilde{\lambda}\I \prec 2(\A + \tilde{\lambda}\I)
\end{align*}
and for the other side, using that $\|\hat{\A}_{\nys} - \A\| \leq \lambda_0$ we have
\begin{align*}
2\M - (\A + \tilde{\lambda}\I) = 2\hat{\A}_{\nys} - \A + \tilde{\lambda}\I  \succeq \hat{\A}_{\nys} - \A + \lambda_0 \I \succeq \mathbf{0}
\end{align*}
which together give $\M \approx_2 \A+\tilde{\lambda}\I$. Thus we can use Lemma~\ref{lem:trans_variable} and have Eq.\eqref{eq:y_to_z} holds.
Now given any $\epsilon>0$, if in Eq.\eqref{eq:y_to_z} we set $\epsilon_1 = O(\epsilon_0/\kappa^{3/2})$ for $\epsilon_0 = O(\frac{1}{(\kappa_{\M}\cdot n/\epsilon)^c})$, then by Lemma~\ref{lem:stable_lanczos}, after $t_{\max} = O(\sqrt{\kappa_{\M}} \log (\kappa_{\M} /\epsilon))$ iterations, preconditioned Lanczos outputs $\tilde{\x}$ such that
\begin{align*}
\|\tilde{\x} - \x^*\|_{\A+\lambda\I} \leq \epsilon  \cdot \|\x^*\|_{\A+\lambda\I}.
\end{align*}
Thus, we obtain the final convergence result. Next, we analyze the second level.

\paragraph{Part 3: Level-2 Preconditioning.}
For the second level we use preconditioned Lanczos to solve the $s\times s$ linear system $(\C^\top\C+\tilde{\lambda}\W)\y = \C^\top\r$.  As for the preconditioner we choose $\M_2 = \tilde{\C}^\top\tilde{\C} + \tilde{\lambda}\W = \C^\top\mathbf{\Phi}^\top\mathbf{\Phi}\C + \tilde{\lambda}\W$, where $\mathbf{\Phi} \in \R^{\phi\times n}$ is a sparse sketching matrix constructed using standard subspace embedding techniques \cite{chenakkod2023optimal} (see Lemma~\ref{lem:precondition_sparse}) with choice $\phi = O(s+\log 1/\delta)$. By Lemma~\ref{lem:precondition_sparse}, with probability $1-\delta/2$ we have $\kappa_{\M_2} \coloneqq \kappa(\M_2^{-1/2}(\C^\top\C+\tilde{\lambda}\W)\M_2^{-1/2}) = O(1)$. Since in Algorithm~\ref{alg:msp_psd} we pre-compute $\M_2^{-1}$ exactly, the error caused by the preconditioner is $0$. By using \Cref{thm:stable_lanczos} we know that after $O(\sqrt{\kappa_{\M_2}} \log (\kappa_{\M_2}/\epsilon_1)) = O(\log 1 /\epsilon_1) = O(\log (\kappa / \epsilon_0))$ iterations, preconditioned Lanczos will output $\hat{\y}$ that satisfies
\begin{align*}
\|\hat{\y} - \y^*\|_{\C^\top\C+\tilde{\lambda}\W} \leq \epsilon_1\cdot \|\y^*\|_{\C^\top\C+\tilde{\lambda}\W}
\end{align*}
which is exactly the assumption we need in Level-1. By applying a union bound we finish the proof.

\paragraph{Part 4: Total Cost.}
Finally we consider the total cost of Algorithm~\ref{alg:msp_psd}. First we pre-compute $\C = \A\S^\top$ and $\W = \S\C$ which takes $O(\gamma\cdot\nnz(\A) + \gamma\cdot\nnz(\C))$ time. Notice that $\gamma = O(\log(l/\delta))$ and $\nnz(\C) \leq \gamma\cdot \nnz(\A)$, thus this step takes $O(\nnz(\A) \log^2(l/\delta))$ time. Next, we pre-compute $\tilde{\C} = \mathbf{\Phi}\C$ and $\M_2^{-1} = (\tilde{\C}^\top\tilde{\C} + \tilde{\lambda}\W)^{-1}$ using Lemma~\ref{lem:precondition_sparse}, which takes $O(\nnz(\C)\log(s/\delta) + s^2 \log^4 (s/\delta) + s^{\omega}) = O(\nnz(\A) \log^2(l/\delta) + l^{\omega}\log^{\omega}(l/\delta))$. By Theorem~\ref{thm:stable_lanczos} and above analysis, the number of iterations for Level-1 preconditioned Lanczos is $t_{\max} = O(\sqrt{\kappa_{\M}} \log(\kappa_{\M}/\epsilon)) = O(\sqrt{\bar{\kappa}_{l,\lambda} n/l} \log(\bar{\kappa}_{l,\lambda} n/\epsilon))$. 

In Level-2, in each iteration we need to compute matrix-vector products of $\C^\top\C+\tilde{\lambda}\W$. Notice that $(\C^\top\C+\tilde{\lambda}\W)\x = (\S\A + \S)\A\S^\top\x$ for any vector $\x$, thus we can first compute $\S^\top\x$ which takes $O(n \log(l/\delta))$, then compute $\A(\S^\top\x)$ which takes $O(\nnz(\A))$, and finally apply $\S\A$ and $\S$ to this vector which takes $O(n\log(l/\delta) + \nnz(\A))$ time. Since there are $O(\log 1/\epsilon_1) = O(\log (\kappa / \epsilon_0))$ iterations of preconditioned Lanczos for Level-2, we conclude that the total cost is
\begin{align*}
    & ~ O\left(\nnz(\A)\log^2(l/\delta) + l^{\omega}\log^{\omega}(l/\delta) + \nnz(\A)\log(l/\delta)\log(\kappa/\epsilon_0)\cdot \sqrt{\bar{\kappa}_{l,\lambda} n/l}\log(\bar{\kappa}_{l,\lambda}n /\epsilon)\right) \\
    = & ~ O\left(\nnz(\A)\sqrt{\bar{\kappa}_{l,\lambda} n/l} \log^2(l/\delta)\log(\kappa/\epsilon_0)\log(\bar{\kappa}_{l,\lambda}n/\epsilon) + l^{\omega} \log^{\omega}(l/\delta)\right) \\
    = & ~ \tilde{O}\left(\nnz(\A)\sqrt{\bar{\kappa}_{l,\lambda} n / l} \log^2(\kappa/\epsilon) + l^{\omega}\right).
\end{align*}

\subsection{Coarse Nystr\"om Preconditioner Based on Sparse Embedding}
\label{sec:sparse_embed}
In this section we give a proof of Lemma~\ref{lem:precondition_sparse_main}, which measures the quality of Nystr\"om preconditioner based on a sparse embedding matrix. Notice that, compared to the most commonly used guarantee that the preconditioned linear system has a constant condition number, we are using a ``coarse'' Nystr\"om preconditioner which benefits us by reducing the cost of constructing it, while slightly sacrificing its quality and allowing the condition number after preconditioning to be ``not so small''.

To characterize the properties of sparse embedding matrix, we first introduce the following notion of oblivious subspace embedding moments property, which can be seen as a generalization of the ``JL-moment property'' introduced by \cite{kane2014sparser}.
\begin{definition}[Definition 4 in \cite{CohenNelsonWoodruff:2016}]\label{def:ose}
A distribution $\mathcal{D}$ over $\R^{s\times n}$ has $(\eps, \delta, d, \ell)$-OSE moments if for all matrices $\U\in\R^{n\times d}$ with orthonormal columns, we have
\begin{align*}
\E_{\bpi\sim\mathcal{D}} \|(\bpi\U)^\top(\bpi\U) - \I\|^{\ell} < \epsilon^{\ell} \cdot \delta.
\end{align*}
\end{definition}
\cite{cohen2016nearly} showed that a sparse embedding matrix (Definition \ref{def:sparse_embed}) with sketch size $s = O(d\log (d/\delta) /\epsilon^2)$ and $\gamma = O(\log(d/\delta) / \epsilon)$ non-zeros per column satisfies $(\epsilon, \delta, d, \log(d/\delta))$-OSE moments property, as stated in the following lemma.
\begin{lemma}[Theorem 4.2 in \cite{cohen2016nearly}]\label{lem:cohen_bound}
For any $\epsilon, \delta \in (0,1/2)$ and $B>2$, a sparse subspace embedding matrix $\S$ with $s = O(Bd\log (d/\delta) /\epsilon^2)$ and $\gamma = O(\log_B(d/\delta) /\epsilon)$ satisfies
\begin{align*}
\E\,\|(\S\U)^\top(\S\U) - \I\|^{\log(d/\delta)} < \epsilon^{\log(d/\delta)}\cdot\delta.
\end{align*}
\end{lemma}

\cite{CohenNelsonWoodruff:2016} showed that the aforementioned OSE moment property is a sufficient condition for the approximated matrix multiplication (AMM), see Lemma~\ref{lem:OSE_req}. Moreover, as a corollary they provided a low rank approximation result based on the AMM property, see Lemma~\ref{lem:sparse_embed}.
\begin{lemma}[Theorem 6 in \cite{CohenNelsonWoodruff:2016}]\label{lem:OSE_req} 
Given $k, \epsilon, \delta \in (0, 1/2)$, let $\mathcal{D}$ be any distribution over matrices with $n$ columns with the $(\eps, \delta, 2k, \ell)$-OSE moment property for some $\ell \geq 2$, then for any $\A,\B$ we have
\begin{align*}
\Pr_{\bpi \sim \mathcal{D}}\left(\left\|(\bpi\A)^\top(\bpi\B) - \A^\top\B \right\| \geq \epsilon\sqrt{(\|\A\|^2 + \|\A\|^2_F/k)(\|\B\|^2 + \|\B\|^2_F/k)}\right) < \delta.
\end{align*}
\end{lemma}

\begin{lemma}[Theorem 8 in \cite{CohenNelsonWoodruff:2016}]\label{lem:sparse_embed} 
For matrix $\A\in\R^{n\times d}$, let $\A = \U\mathbf{\Sigma}\V^\top$ be its SVD and $\A_k = \U_k\mathbf{\Sigma}_k \V_k^\top$ be its best rank-$k$ approximation. If $\S\in\R^{s\times n}$ satisfies the following two properties:
\begin{enumerate}
    \item approximate spectral norm matrix multiplication w.r.t. $\U_k$ and $\A_{\bar{k}}\coloneqq \A-\A_k$, that is:
    \begin{align*}
    \|\U_k^\top\S^\top\S\A_{\bar{k}} - \U_k^\top\A_{\bar{k}}\|^2 \leq \frac{\epsilon}{8}\left(\|\U_k\|^2 + \frac{\|\U_k\|_F^2}{k}\right)\left(\|\A_{\bar{k}}\|^2 + \frac{\|\A_{\bar{k}}\|_F^2}{k}\right)
    \end{align*}
    \item $\S$ is a $1/2$-subspace embedding for the column space of $\U_k$,
\end{enumerate}
then, defining $\P_{\S}$ as the orthogonal projection onto the row-space of $\S\A$, and denoting $\tilde{\A}_k$ as the best rank-$k$ approximation of $\tilde{\A} = \A\P_{\S}$, we have:
\begin{align*}
\|\A - \tilde{\A}_k\|^2 \leq (1+\epsilon) \|\A-\A_k\|^2 + (\epsilon/k)\cdot \|\A-\A_k\|_F^2.
\end{align*}
\end{lemma}
In order to use Lemma~\ref{lem:sparse_embed} to get a low-rank approximation result, there are two requirements that we need to satisfy. As for the first approximate spectral norm matrix multiplication assumption, it follows directly since the OSE moment property holds; as for the second subspace embedding property, it can be shown from the definition of sparse embedding matrix. Formally, we give the proof of Lemma~\ref{lem:precondition_sparse_main} as follows.

\begin{proof}[Proof of Lemma~\ref{lem:precondition_sparse_main}]
For positive definite matrix $\A$, we denote $\Z\coloneqq\A^{1/2}$ be such that $\A = \Z\Z^\top$. Given $\delta <1/2$ and $l<n$, we first verify that the sparse embedding matrix $\S\in\R^{s\times n}$ with $s = O(l \log(l/\delta))$ and $\gamma = O(\log(l/\delta))$ satisfies the two requirements of Lemma~\ref{lem:sparse_embed} on matrix $\Z$:
\begin{itemize}
    \item Requirement 1: we use Lemma~\ref{lem:cohen_bound} with choice $d=l$ and $\epsilon=1/2\sqrt{2}$, thus the sparse embedding matrix with $s = O(l\log(l/\delta))$ and $\gamma = O(\log(l/\delta))$ satisfies the $(1/2\sqrt{2}, \delta/2, l, \log (l/\delta))$-OSE moment property. By using Lemma~\ref{lem:OSE_req} to matrix $\U_k$ and $\A_{\bar{k}}$, we can show the first requirement holds to $\epsilon=1$ with probability at least $1-\delta/2$.
    \item Requirement 2: we also use Lemma~\ref{lem:cohen_bound} with Markov's inequality, which shows that $\S$ satisfies
    \begin{align*}
    \|(\S\U_l)^\top(\S\U_l) - \U_l^\top\U_l\| \leq \frac{1}{2\sqrt{2}} < \frac{1}{2}
    \end{align*}
    with probability at least $1-\delta/2$. Thus we have the $1/2$-subspace embedding property.
\end{itemize}
With the above two requirements being satisfied, we apply Lemma~\ref{lem:sparse_embed} with choice $\epsilon=1$ to matrix $\Z$, and have the following holds with probability $1-\delta$:
\begin{align}\label{eq:sparse_precondition_bound}
\|\Z - \tilde{\Z}_l\|^2 \leq 2\sigma_{l+1}^2 + \frac{1}{l}\sum_{i\geq l+1} \sigma_i^2
\end{align}
where $\{\sigma_i\}_{i=1}^n$ are the singular values of $\Z$ in decreasing order. Notice that here $\tilde{\Z} = \Z\P_{\S}$ and $\tilde{\Z}_l$ is the best rank-$l$ approximation of $\tilde{\Z}$, thus by denoting $\P_l = \V_l\V_l^\top$, we have $\tilde{\Z}_l = \tilde{\Z}\P_l = \Z\P_{\S}\P_l$. Moreover, in Lemma~\ref{lem:project_twice_order} we show that $\|\Z - \tilde{\Z}_l\| \geq \|\Z - \tilde{\Z}\|$, by combining this result and Eq.\eqref{eq:sparse_precondition_bound} we have an error bound on the term $\|\Z - \tilde{\Z}\|^2 = \|\Z-\Z\P_{\S}\|^2$. Further notice the following:
\begin{align*}
\left\|\Z\left(\I - \P_{\S}\right)\right\|^2 = \big\|\Z\left(\I - \P_{\S}\right)\Z^\top\big\|
= \big\|\underbrace{\Z\Z^\top}_{\A} - \underbrace{\Z\Z^\top\S^\top(\S\Z\Z^\top\S^\top)^{-1}\S\Z\Z^\top}_{\hat{\A}_{\nys}}\big\|.
\end{align*}
Thus, if we denote $\hat{\A}_{\nys} \coloneqq \A\S^\top(\S\A\S^\top)^{-1}\S\A$ as the rank-$s$ Nystr\"om approximation of matrix $\A$, then we can bound the error as 
\begin{align*}
\|(\A+\lambda\I) - (\hat{\A}_{\nys}+\lambda\I)\| = \|\A - \hat{\A}_{\nys}\| \leq 2\sigma_{l+1}^2 + \frac{1}{l} \sum_{i\geq l+1}\sigma_i^2 = 2\lambda_{l+1} + \frac{1}{l} \sum_{i\geq l+1}\lambda_i.
\end{align*}
Based on this result, if we construct the preconditioner as $\M = \hat{\A}_{\nys} + \lambda\I + \lambda_0\I$ where $\lambda_0 \geq 2\lambda_{2l+1} + \frac{1}{2l} \sum_{i\geq 2l+1}\lambda_i$ and $\hat{\A}_{\nys}$ is the rank-$2s$ (instead of rank-$s$) Nystr\"om approximation, then we have $\A+\lambda\I \preceq \M \preceq (\A + \lambda\I) + \lambda_0\I$, which gives $\kappa_{\M} \leq 1 + \frac{\lambda_0}{\lambda_n+\lambda}$. Further notice that we can bound this term as
\begin{align*}
2\lambda_{2l+1} + \frac{1}{2l} \sum_{i\geq 2l+1}\lambda_i \leq \frac{2}{l}\sum_{i=l+1}^{2l}\lambda_i + \frac{1}{2l} \sum_{i\geq 2l+1}\lambda_i \leq \frac{2}{l}\sum_{i\geq l+1}\lambda_i,
\end{align*}
thus by setting $\lambda_0 = \frac{2}{l}\sum_{i\geq l+1}\lambda_i$ we have the following holds for some $C=O(1)$:
\begin{align}\label{eq:nys_precondition}
\kappa_{\M} \leq 1 + \frac{\lambda_0}{\lambda_n+\lambda} \leq \frac{C}{l}\sum_{i\geq l+1}\frac{\lambda_i}{\lambda_n+\lambda} \leq \frac{C\bar{\kappa}_{l,\lambda} n}{l}.
\end{align}
\end{proof}
\begin{lemma}\label{lem:project_twice_order}
We have $\P_{\S}\P_l = \P_l$, which gives $\tilde{\Z}_l = \Z\P_l$. We further have
\begin{align*}
\|\Z - \tilde{\Z}_l\| \geq \|\Z - \tilde{\Z}\|.
\end{align*}
\end{lemma}
\begin{proof}
Let $\tilde{\Z} = \tilde{\U}\tilde{\mathbf{\Sigma}}\tilde{\V}^\top$ be its SVD, then according to definition we have $\tilde{\Z}_l = \tilde{\U}_l\tilde{\mathbf{\Sigma}}_l\tilde{\V}_l^\top$, and we can write the projection matrix $\P_l$ as $\P_l = \tilde{\V}_l\tilde{\V}_l^\top$ where $\tilde{\V}_l\in\R^{n \times l}$. With these, we can always write $\P_{\S}$ as $\P_{\S} = \Q\Q^\top$ where $\Q \in\R^{n\times s}$ satisfies $\tilde{\V}_l = \Q \I_l$ for some $\I_l\in\R^{s\times l}$ satisfying $\I_l^\top\I_l = \I$. Thus we have
\begin{align*}
\P_{\S}\P_l = \Q\Q^\top\tilde{\V}_l\tilde{\V}_l^\top = \Q\Q^\top\Q\I_l\I_l^\top\Q^\top = \Q\I_l\I_l^\top\Q^\top = \tilde{\V}_l\tilde{\V}_l^\top= \P_l.
\end{align*}
\end{proof}

\subsection{Variable Transformation}
\label{sec:trans_variable}
In this section we give a proof of Lemma~\ref{lem:trans_variable}.
\begin{proof}[Proof of Lemma~\ref{lem:trans_variable}]
We first show the relation between the two aforementioned linear systems $(\C^\top\C+\tilde{\lambda}\W)\y = \C^\top\r$ and $\M\w = \r$. Notice that 
\begin{align*}
& ~ \M\cdot\left(\I - \C(\C^\top\C+\tilde{\lambda}\W)^{-1}\C^\top\right) \\
= & ~ \left(\C\W^{-1}\C^\top + \tilde{\lambda}\I\right)\cdot\left(\I - \C(\C^\top\C+\tilde{\lambda}\W)^{-1}\C^\top\right) \\
= & ~ \C\W^{-1}\C^\top - \C\W^{-1}\C^\top \C(\C^\top\C+\tilde{\lambda}\W)^{-1}\C^\top + \tilde{\lambda}\I - \tilde{\lambda} \C(\C^\top\C+\tilde{\lambda}\W)^{-1}\C^\top \\
= & ~ \C\W^{-1}\C^\top - \C\W^{-1}\C^\top + \tilde{\lambda}\C(\C^\top\C+\tilde{\lambda}\W)^{-1}\C^\top + \tilde{\lambda}\I - \tilde{\lambda} \C(\C^\top\C+\tilde{\lambda}\W)^{-1}\C^\top  \\
= & ~ \tilde{\lambda}\I
\end{align*}
which gives 
\begin{align}\label{eq:M_inverse}
\M^{-1} = \frac{1}{\tilde{\lambda}}\left(\I - \C(\C^\top\C+\tilde{\lambda}\W)^{-1}\C^\top\right).
\end{align}
Thus, if we denote $\y^* = (\C^\top\C+\tilde{\lambda}\W)^{-1}\C^\top\r$ as the solution of the first linear system, then the solution of the second system can be expressed as $\w^* = \frac{1}{\tilde{\lambda}}(\r - \C\y^*)$, and we further have $\tilde{\lambda} (\hat{\w} - \w^*) = \C(\y - \hat{\y})$. Furthermore, by assumption $\M = \hat{\A}_{\nys} +\tilde{\lambda}\I \approx_2 \A+\tilde{\lambda}\I$ we have
\begin{align*}
\C^\top\M\C \approx_2 \C^\top(\A+\tilde{\lambda}\I)\C = \S\A(\A+\tilde{\lambda}\I)\A\S^\top.
\end{align*}
Based on this observation and the assumption that $\|\hat{\y} - \y^*\|_{\C^\top\C+\tilde{\lambda}\W}\leq \epsilon_1 \|\y^*\|_{\C^\top\C+\tilde{\lambda}\W}$, we have
\begin{align*}
\tilde{\lambda}^2\|\hat{\w} - \w^*\|_{\M}^2 = & ~ \|\C(\hat{\y} - \y^*)\|_{\M}^2 = \|\hat{\y} - \y^*\|_{\C^\top\M\C}^2 \leq 2\|\hat{\y} - \y^*\|_{\S\A(\A+\tilde{\lambda}\I)\A\S^\top}^2 \\
= & ~ 2\|\A\S^\top(\hat{\y} - \y^*)\|_{\A+\tilde{\lambda}\I}^2 \leq 2\|\A^{1/2}\|^2 \cdot \|\A^{1/2}\S^\top(\hat{\y} - \y^*)\|_{\A+\tilde{\lambda}\I}^2 \\
= & ~ 2\|\A\|\cdot \|\hat{\y} - \y^*\|_{\S(\A^2+\tilde{\lambda}\A)\S^\top}^2 = 2\|\A\|\cdot \|\hat{\y} - \y^*\|_{\C^\top\C+\tilde{\lambda}\W}^2 \\
\leq & ~ 2\epsilon_1^2\|\A\| \cdot \|\y^*\|_{\C^\top\C+\tilde{\lambda}\W}^2 = 2\epsilon_1^2\|\A\| \cdot \|\A^{1/2}\S^\top\y^*\|_{\A+\tilde{\lambda}\I}^2 \\
\leq & ~ 2\epsilon_1^2 \kappa(\A) \cdot \|\A\S^\top\y^*\|_{\A+\tilde{\lambda}\I}^2  = 2\epsilon_1^2 \kappa(\A) \cdot \|\y^*\|_{\S\A(\A+\tilde{\lambda}\I)\A\S^\top}^2 \\
\leq & ~ 4\epsilon_1^2 \kappa(\A) \cdot \|\y^*\|_{\C^\top\M\C}^2 = 4\epsilon_1^2 \kappa(\A) \cdot \|\C\y^*\|_{\M}^2.
\end{align*}
Thus we have
\begin{align}\label{eq:trans_1}
\|\hat{\w} - \w^*\|_{\M} \leq \frac{2\epsilon_1}{\tilde{\lambda}}\sqrt{\kappa(\A)} \cdot \|\C\y^*\|_{\M}
\end{align}
Furthermore, since $\r = \M\w^* = \frac{1}{\tilde{\lambda}}\M(\r - \C\y^*)$, we have
\begin{align}\label{eq:trans_2}
\|\r\|_{\M} =\|\M^{1/2}\r\| \leq \|\M\|\|\M^{-1/2}\r\| = \frac{1}{\tilde{\lambda}}\|\M\| \|\M^{1/2}(\r-\C\y^*)\| = \frac{1}{\tilde{\lambda}} \|\M\| \|\r-\C\y^*\|_{\M}.
\end{align}
By combining Eq.\eqref{eq:trans_1} and \eqref{eq:trans_2}, we have the following holds for some constant $c_0>4$:
\begin{align*}
\|\hat{\w} - \w^*\|_{\M} \leq & ~ \frac{2\epsilon_1\kappa^{1/2}}{\tilde{\lambda}} \|\C\y^*\|_{\M}
\leq \frac{2\epsilon_1\kappa^{1/2}}{\tilde{\lambda}} (\|\r-\C\y^*\|_{\M} + \|\r\|_{\M}) \\
\leq & ~ 2\epsilon_1\kappa^{1/2}\left(\frac{\|\M\|}{\tilde{\lambda}} + 1\right)\cdot \frac{1}{\tilde{\lambda}}\|\r-\C\y^*\|_{\M} = 2\epsilon_1\kappa^{1/2}\frac{\|\M+\tilde{\lambda}\I\|}{\tilde{\lambda}}\cdot \|\hat{\w}\|_{\M} \\
\leq & ~ 4\epsilon_1\kappa^{1/2} \frac{\|\A+\tilde{\lambda}\I\|}{\tilde{\lambda}}\cdot\|\hat{\w}\|_{\M} \leq  c_0\epsilon_1\kappa^{3/2} \|\hat{\w}\|_{\M}
\end{align*}
where $\kappa = \kappa(\A)$. Thus we finish the proof.
\end{proof}

\section{Three-Level \AlgName{} for Solving General Linear Systems} \label{sec:proof_rec}
With the easier to analyze PD case being proved in Section~\ref{s:psd-proof}, in this section we give a proof of Theorem~\ref{thm:main_rec}, where given $\A\in\R^{m\times n}$ with full column rank, $\c\in\R^n$ and $\lambda\geq 0$, we want to solve $(\A^\top\A+\lambda\I) \x = \c$. As we discuss in Section~\ref{sec:main_results}, such a linear system is more general in the sense that it can not only recover the rectangular linear system, but also have applications to kernel ridge regression and Schatten norm estimation.

Since $\A^\top\A+\lambda\I \in \mathcal{S}_n^{++}$, one natural idea is to construct the Nystr\"om preconditioner similar as before. Given sketching matrix $\S\in\R^{s\times n}$, define $\tilde{\A} \coloneqq \A\S^\top, \C \coloneqq \A^\top\A\S^\top = \A^\top\tilde{\A}$ and $\W \coloneqq \S\A^\top\A\S^\top = \tilde{\A}^\top\tilde{\A}$. Then, $\C\W^{-1}\C^\top$ is the classical Nystr\"om preconditioner of matrix $\A^\top\A$. Let $\{\sigma_i\}_{i=1}^n$ be the singular values of $\A$ in decreasing order. By using Lemma~\ref{lem:precondition_sparse_main}, we can again show that the matrix $\M \coloneqq \C\W^{-1}\C^\top + \tilde{\lambda}\I$ where $\tilde{\lambda} \coloneqq \lambda + \frac{2}{l}\sum_{i>l}\sigma_i^2$, is a useful preconditioner for this problem, and we also have the inversion formula:
\begin{align*}
\M^{-1} = \frac{1}{\tilde{\lambda}} \left(\I - \C(\C^\top\C + \tilde{\lambda}\W)^{-1}\C^\top\right).
\end{align*}
However, different from the positive definite case where we can compute $\C$ and $\W$ explicitly, in this case constructing $\C$ and $\W$ takes $\tilde{O}(s\cdot \nnz(\A))$ time which is much too expensive. To avoid computing $\C$, we first observe that $\tilde{\A}\tilde{\A}^\top+\tilde{\lambda}\I \approx_{2} \A\A^\top+\tilde{\lambda}\I$ (which is proved in Lemma~\ref{lem:spec_approx}), concluding that
\begin{align*}
\W^2 + \tilde{\lambda}\W = \tilde{\A}^\top (\tilde{\A}\tilde{\A}^\top + \tilde{\lambda}\I)\tilde{\A} \approx_{2} \tilde{\A}^\top(\A\A^\top+\tilde{\lambda}\I)\tilde{\A} = \C^\top\C+\tilde{\lambda}\W.
\end{align*}
With this observation, we do not need to compute $\C$, and we can use $\M_2 \coloneqq \W^2+\tilde{\lambda}\W$ as the Level-2 preconditioner for solving the linear system $(\C^\top\C+\tilde{\lambda}\W)\y = \r$. However we still cannot afford to compute $\W$, not to mention $\M_2$. To address this, notice that
\begin{align*}
\M_2^{-1} = (\W^2+\tilde{\lambda}\W)^{-1} = & ~ (\tilde{\A}^\top\tilde{\A})^{-1}(\tilde{\A}^\top\tilde{\A}+\tilde{\lambda}\I)^{-1} \\
= & ~ (\tilde{\A}^\top\tilde{\A})^{-1} \frac{1}{\tilde{\lambda}} \left(\I - \tilde{\A}^\top\tilde{\A}(\tilde{\A}^\top\tilde{\A}+\tilde{\lambda}\I)^{-1}\right) \\
= & ~ \frac{1}{\tilde{\lambda}}\left((\tilde{\A}^\top\tilde{\A})^{-1} - (\tilde{\A}^\top\tilde{\A} + \tilde{\lambda}\I)^{-1}\right).
\end{align*}
Since we only need to have access to $\M_2^{-1}\r$ in Level-2 preconditioned Lanczos, it suffices to compute $(\tilde{\A}^\top\tilde{\A})^{-1}\r$ and $(\tilde{\A}^\top\tilde{\A} + \lambda\I)^{-1}\r$. To solve this, we introduce the third level, where we again use preconditioned Lanczos to solve two positive definite linear systems $(\tilde{\A}^\top\tilde{\A})\u = \r$ and $(\tilde{\A}^\top\tilde{\A} + \tilde{\lambda}\I)\v = \r$. To construct the Level-3 preconditioners, notice that $\tilde{\A}$ is an $m \times s$ tall matrix, and we can thus construct an ${O}(s) \times s$ sketch of it.

Let $\mathbf{\Phi} \in \R^{\phi \times m}$ be the oblivious embedding matrix in Lemma~\ref{lem:precondition_sparse} with   $\phi = {O}(s)$. We first compute $\hat{\A} = \mathbf{\Phi}\tilde{\A} \in\R^{\phi\times s}$, then construct the preconditioner $\M_{3a}^{-1} \coloneqq \hat{\A}^\top\hat{\A}$ and $\M_{3b}^{-1} \coloneqq \hat{\A}^\top\hat{\A} + \tilde{\lambda}\I$. Lemma~\ref{lem:precondition_sparse} guarantees that both preconditioned linear systems have a constant condition number, thus solving them using preconditioned Lanczos only takes $O(\log 1/\epsilon_2)$ steps for the targeted accurancy $\epsilon_2$. Moreover, since both $\M_{3a}$ and $\M_{3b}$ are $s \times s$ positive definite matrices, we can afford to compute their inverse exactly, and the errors caused by the Level-3 preconditioners are $0$. By going back from Level-3 $\to$ Level-2 $\to$ Level-1 and analyzing the error bounds carefully, we can finally arrive at the convergence result on the original linear system.

Our algorithms are as follows. Algorithm~\ref{alg:msp_rec_main} is the main algorithm for solving the linear system $(\A^\top\A+\lambda\I)\x = \c$, in which we run preconditioned Lanczos with Level-1 preconditioner $\M$ defined by function $\mathsf{SolverM1}$. Algorithm~\ref{alg:msp_rec_func_outer} defines function $\mathsf{SolverM1}$ by solving linear system $\M\w = \r$ with preconditioned Lanczos, where the Level-2 preconditioner $\M_2$ is defined by function $\mathsf{SolverM2}$. Lastly, Algorithm~\ref{alg:msp_rec_func_inner} defines function $\mathsf{SolveM2}$ by solving two linear systems $(\tilde{\A}^\top\tilde{\A})\u = \r$ and $(\tilde{\A}^\top\tilde{\A} + \tilde{\lambda}\I)\v = \r$ with preconditioned Lanczos respectively, whereas the Level-3 preconditioners $\M_{3a}$ and $\M_{3b}$ are pre-computed (see line 6 of Algorithm~\ref{alg:msp_rec_main}). As a comparison of the quality of the preconditioners we use in different levels, we have:
\begin{align*}
\begin{cases}
\text{Level-1: } \kappa_{\M} \coloneqq \kappa(\M^{-1/2}(\A^\top\A+\lambda\I)\M^{-1/2}) = O(\bar{\kappa}_{l,\lambda}^2 n/l) & \textbf{Coarse} \\
\text{Level-2: } \kappa_{\M_2} \coloneqq \kappa(\M_2^{-1/2}(\C^\top\C+\tilde{\lambda}\W)\M_2^{-1/2}) = O(1) & \textbf{Fine} \\
\text{Level-3a: } \kappa_{\M_{3a}} \coloneqq\kappa(\M_{3a}^{-1/2} (\tilde{\A}^\top\tilde{\A})\M_{3a}^{-1/2}) = O(1) & \textbf{Fine} \\
\text{Level-3b: } \kappa_{\M_{3b}} \coloneqq \kappa (\M_{3b}^{-1/2} (\tilde{\A}^\top\tilde{\A}+\tilde{\lambda}\I) \M_{3b}^{-1/2}) = O(1) & \textbf{Fine} 
\end{cases}
\end{align*}

\begin{algorithm}[!ht]
\caption{\AlgName{} for solving linear system $(\A^\top\A+\lambda\I)\x = \c$.}
\label{alg:msp_rec_main}
\begin{algorithmic}[1]
\State \textbf{Input: }matrix $\A\in\R^{m\times n}$, regularize term $\lambda$ and $\lambda_0$, vector $\c \in \R^n$, sparse sketch size $s$, \# of non-zeros $\gamma$, level-2 sketch size $\phi$, \# of iterations $t_{\max}$;
\State Compute $\tilde{\lambda} = \lambda+\lambda_0$;
\State Construct sparse embedding matrix $\S\in\R^{s\times n}$ with $\gamma$ non-zeros per column; \Comment{Def.~\ref{def:sparse_embed}.}
\State Compute $\tilde{\A} = \A\S^\top$ and $\hat{\A} = \mathbf{\Phi}\tilde{\A} \in \R^{\phi \times s}$; \Comment{Lemma~\ref{lem:precondition_sparse}.}
\State Compute $\M_{3a} = \hat{\A}^\top\hat{\A}$ and $\M_{3b} = \M_{3a} + \tilde{\lambda}\I$
\State Compute $\M_{3a}^{-1}$ and $\M_{3b}^{-1}$;
\State Solve $\tilde{\x}$ by preconditioned Lanczos with $(\A^\top\A+\lambda\I, \c, \mathsf{SolveM1}, t_{\max})$; \Comment{Alg.~\ref{alg:lanczos_left_precon}.}
\\
\Return $\tilde{\x}$; \Comment{Solves $(\A^\top\A+\lambda\I)\x = \c$.}
\end{algorithmic}
\end{algorithm}

\begin{algorithm}[!ht]
\caption{Level-1 auxiliary function $\mathsf{SolveM1}$ for solving $\M\w=\r$.}
\label{alg:msp_rec_func_outer}
\begin{algorithmic}[1]
\State \textbf{function} $\mathsf{SolveM1}(\r)$:
\State \quad Solve $\hat{\y}$ by Lanczos with $(\tilde{\A}^\top\A\A^\top\tilde{\A}+\tilde{\lambda}\tilde{\A}^\top\tilde{\A}, \tilde{\A}^\top\A\r, \mathsf{SolveM2}, O(\log \kappa / \epsilon_0))$; \Comment{Alg.~\ref{alg:lanczos_left_precon}.}
\State \quad Compute $\hat{\w} = \frac{1}{\tilde{\lambda}}(\r - \A^\top\tilde{\A} \hat{\y})$;
\\
\Return $\hat{\w}$; \Comment{Solves $\M\w = \r$ for $\M = \C\W^{-1}\C^\top + \tilde{\lambda}\I$.}
\end{algorithmic}
\end{algorithm}

\begin{algorithm}[!ht]
\caption{Level-2 auxiliary function $\mathsf{SolveM2}$ for solving $\M_2\z = \r$.}
\label{alg:msp_rec_func_inner}
\begin{algorithmic}[1]
\State \textbf{function} $\mathsf{SolveM2}(\r)$:
\State \quad Solve $\hat{\u}$ by preconditioned Lanczos with $(\tilde{\A}^\top\tilde{\A}, \r, \M_{3a}^{-1}, O(\log \kappa l/\epsilon_0) )$; \Comment{Alg.~\ref{alg:lanczos_left_precon}.}
\State \quad Solve $\hat{\v}$ by preconditioned Lanczos with $(\tilde{\A}^\top\tilde{\A}+\tilde{\lambda}\I, \r, \M_{3b}^{-1}, O(\log \kappa l/\epsilon_0) )$; \Comment{Alg.~\ref{alg:lanczos_left_precon}.}
\State \quad Compute $\hat{\z} = \frac{1}{\tilde\lambda}(\hat{\u}-\hat{\v})$;
\\
\Return $\hat{\z}$; \Comment{Solves $\M_2 \z = \r$ for $\M_2 = \W^2 + \tilde{\lambda}\W$.}
\end{algorithmic}
\end{algorithm}

\subsection{Proof of Theorem~\ref{thm:main_rec}}
Given matrix $\A\in\R^{m\times n}$ and $\lambda\geq 0$, we consider to solve the regularized linear system $(\A^\top\A+\lambda\I)\x = \c$ without explicitly computing $\A^\top\A$. Since $\A$ has full column rank, we know $\A^\top\A+\lambda\I\in\R^{n\times n}$ is positive definite. We prove the result in following five parts.

\paragraph{Part 1: Nystr\"om Preconditioner with Sparse Embedding.}
Let $\S\in\R^{s\times n}$ be the sparse embedding matrix according to Lemma~\ref{lem:precondition_sparse_main} with $s = O(l\log(l/\delta))$ and $\gamma = O(\log(l/\delta))$. Denote $\C \coloneqq \A^\top\A\S^\top$ and $\W \coloneqq \S\A^\top\A\S^\top$, then $\C\W^{-1}\C^\top$ is the Nystr\"om approximation of $\A^\top\A$. Denote $\tilde{\lambda} = \lambda+\lambda_0$, let $\M = \C\W^{-1}\C^\top+\tilde{\lambda}\I$ be the preconditioner, then according to Eq.\eqref{eq:M_inverse} in the proof of Lemma~\ref{lem:trans_variable} we have
\begin{align*}
\M^{-1} = \frac{1}{\tilde{\lambda}} \left(\I - \C(\C^\top\C+\tilde{\lambda}\W)^{-1}\C^\top\right).
\end{align*}
Given $0 < \delta <1/8$, by applying Lemma~\ref{lem:precondition_sparse_main} to matrix $\A^\top\A$, if we set $\lambda_0 = \frac{2}{l}\sum_{i\geq l+1} \lambda_i(\A^\top\A) = \frac{2}{l}\sum_{i\geq l+1} \sigma_i^2$, then with probability at least $1-\delta/3$ we have $\kappa_{\M} \leq \frac{C\bar{\kappa}_{l,\lambda}^2 n}{l}$ for some $C=O(1)$.

\paragraph{Part 2: Level-1 Preconditioning.}
Similar with the positive definite case, for the first level we use preconditioned Lanczos method (Algorithm~\ref{alg:lanczos_left_precon}) where in each iteration we need to compute $\w^* \coloneqq\M^{-1}\r$. Similarly, we introduce $\y^* \coloneqq (\C^\top\C+\tilde{\lambda}\W)^{-1}\C^\top\r$ and approximate $\w^*$ by $\hat{\w} = \frac{1}{\tilde{\lambda}}(\r - \C\hat{\y})$. Suppose we can solve the linear system $(\C^\top\C+\tilde{\lambda}\W)\y = \C^\top\r$ and have the guarantee that (which will be proved in Level-2 and Level-3):
\begin{align}\label{eq:level_1_req}
\|\hat{\y} - \y^*\|_{\C^\top\C+\tilde{\lambda}\W} \leq \epsilon_1\cdot \|\y^*\|_{\C^\top\C+\tilde{\lambda}\W}.
\end{align}
Notice that since $\kappa(\A^\top\A) = \kappa^2$, Lemma~\ref{lem:trans_variable} guarantees that 
\begin{align*}
\|\hat{\w} - \M^{-1}\r\|_{\M} \leq c_0\epsilon_1\kappa^3 \|\M^{-1} \r\|_{\M} ~~~\text{for some}~~~c_0=O(1).
\end{align*}
By setting $\epsilon_1 = O(\epsilon_0 / \kappa^3)$ we verify the assumption of Lemma~\ref{lem:stable_lanczos} stated in Eq.\eqref{eq:lanczos_assump}, also note that $\kappa(\A^\top\A +\lambda\I) \leq \kappa^2$. Thus by using Lemma~\ref{lem:stable_lanczos}, after $t_{\max} = O(\sqrt{\kappa_{\M}}\log(\kappa_{\M} / \epsilon))$ iterations of preconditioned Lanczos, we can obtain $\tilde{\x}$ such that
\begin{align*}
\|\tilde{\x} - \x^*\|_{\A^\top\A+\lambda\I} \leq \epsilon\cdot \|\x^*\|_{\A^\top\A+\lambda\I}.
\end{align*}
Thus we obtain the final convergence result we want. Next, we analyze the second level.

\paragraph{Part 3: Level-2 Preconditioning.}
Compared with the positive definite case in Section~\ref{s:psd-proof}, the difference here is that we cannot afford to compute $\C$ and $\W$, which makes the analysis more difficult. To address this, we first show in the following Lemma~\ref{lem:spec_approx} that the sketch $\tilde{\A} = \A\S^\top$ satisfies $\tilde{\A}\tilde{\A}^\top+\tilde{\lambda}\I \approx_2\A\A^\top+\tilde{\lambda}\I$ holds with probability $1 - \delta/3$. 
\begin{lemma}[Spectral approximation]\label{lem:spec_approx}
Given matrix $\A\in\R^{m\times n}$ with full column rank, $\lambda\geq 0$ and $0<\delta<1/2$. Given $\log n < l<n$, let $\lambda_0 = \frac{2}{l}\sum_{i>l}\sigma_i^2$ and denote $\tilde{\lambda} = \lambda + \lambda_0$. Let $\S\in\R^{n \times s}$ be the sparse subspace embedding matrix with $s = O(l\log (l/\delta))$ and each column having $\gamma = O(\log(l/\delta))$ non-zeros. Denote $\tilde{\A} = \A\S^\top$, then with probability at least $1-\delta$ we have
\begin{align*}
\tilde{\A}\tilde{\A}^\top+\tilde{\lambda}\I \approx_2 \A\A^\top+\tilde{\lambda}\I.
\end{align*}
\end{lemma}

\noindent
Conditioned on this event, we have
\begin{align*}
\W^2+\tilde{\lambda}\W = \tilde{\A}^\top(\tilde{\A}\tilde{\A}^\top+\tilde{\lambda}\I)\tilde{\A}\approx_2 \tilde{\A}^\top(\A\A^\top+\tilde{\lambda}\I)\tilde{\A} = \C^\top\C+\tilde{\lambda}\W.
\end{align*}
With this result, to solve the linear system $(\C^\top\C+\tilde{\lambda}\W)\y = \C^\top\r$, we can use preconditioned Lanczos with $\M_2 \coloneqq \W^2+\tilde{\lambda}\W$ as the preconditioner, and have $\kappa_{\M_2} = \kappa(\M_2^{-1/2} (\C^\top\C+\tilde{\lambda}\W) \M_2^{-1/2}) \leq 4$. However, we still can not construct $\M_2$ since we cannot afford to compute $\W$. Instead of constructing $\M_2$ explicitly, in each iteration we only need access to $\M_2^{-1}\r$ for a given vector $\r$. Further notice that $\M_2^{-1}$ can be expressed as follows:
\begin{align*}
(\W^2+\tilde{\lambda}\W)^{-1} = & ~ (\tilde{\A}^\top\tilde{\A})^{-1}(\tilde{\A}^\top\tilde{\A}+\tilde{\lambda}\I)^{-1} \\
= & ~ (\tilde{\A}^\top\tilde{\A})^{-1} \frac{1}{\tilde{\lambda}} \left(\I - \tilde{\A}^\top\tilde{\A}(\tilde{\A}^\top\tilde{\A}+\tilde{\lambda}\I)^{-1}\right) \\
= & ~ \frac{1}{\tilde{\lambda}}\left((\tilde{\A}^\top\tilde{\A})^{-1} - (\tilde{\A}^\top\tilde{\A} + \tilde{\lambda}\I)^{-1}\right),
\end{align*}
thus in order to compute $\M_2^{-1}\r$, it suffices to compute $\u^* \coloneqq (\tilde{\A}^\top\tilde{\A})^{-1}\r$ and $\v^* \coloneqq (\tilde{\A}^\top\tilde{\A} + \tilde{\lambda}\I)^{-1}\r$. Suppose we can solve two linear systems $(\tilde{\A}^\top\tilde{\A})\u = \r$, $(\tilde{\A}^\top\tilde{\A} +\tilde{\lambda}\I)\v = \r$ approximately, and have the following guarantee (which will be proved in Level-3):
\begin{align*}
\begin{cases}
\|\hat{\u} - \u^*\|_{\tilde{\A}^\top\tilde{\A}} \leq \epsilon_{2,1} \cdot \|\u^*\|_{\tilde{\A}^\top\tilde{\A}}; \\
\|\hat{\v} - \v^*\|_{\tilde{\A}^\top\tilde{\A}+\tilde{\lambda}\I} \leq \epsilon_{2,2} \cdot \|\v^*\|_{\tilde{\A}^\top\tilde{\A}+\tilde{\lambda}\I}.
\end{cases}
\end{align*}
Denote $\z^* \coloneqq \M_2^{-1}\r = \frac{1}{\tilde{\lambda}}(\u^* - \v^*)$ as the quantity we want to approximate, and denote $\hat{\z} \coloneqq \frac{1}{\tilde{\lambda}}(\hat{\u} - \hat{\v})$.
Observe that $\W^2 +\tilde{\lambda}\W = \tilde{\A}^\top\tilde{\A}\tilde{\A}^\top\tilde{\A} + \tilde{\lambda}\tilde{\A}^\top\tilde{\A} = \tilde{\A}^\top(\tilde{\A}\tilde{\A}^\top + \tilde{\lambda}\I)\tilde{\A}$, then we have
\begin{align*}
\tilde{\lambda} \|\hat{\z} - \z^*\|_{\W^2+\tilde{\lambda}\W} = & ~ \left\|(\hat{\u} - \hat{\v}) - (\u^* - \v^*) \right\|_{\W^2+\tilde{\lambda}\W} \\
\leq & ~ \|\hat{\u} - \u^*\|_{\W^2 + \tilde{\lambda}\W} + \|\hat{\v} - \v^*\|_{\W^2+ \tilde{\lambda}\W} \\
= & ~ \big\|\tilde{\A}(\hat{\u} - \u^*)\big\|_{\tilde{\A}\tilde{\A}^\top+\tilde{\lambda}\I} + \big\|\tilde{\A}(\hat{\v} - \v^*)\big\|_{\tilde{\A}\tilde{\A}^\top+\tilde{\lambda}\I} \\
\leq & ~ \|\tilde{\A}\tilde{\A}^\top+\tilde{\lambda}\I\|^{1/2}\cdot \big\|\tilde{\A}(\hat{\u} - \u^*)\big\| + \big\|\tilde{\A}(\tilde{\A}^\top\tilde{\A}+\tilde{\lambda}\I)^{1/2}(\hat{\v} - \v^*)\big\| \\
\leq & ~ \|\tilde{\A}\tilde{\A}^\top+\tilde{\lambda}\I\|^{1/2}\cdot \epsilon_{2,1} \|\tilde{\A}\u^*\| + \|\tilde{\A}^\top\tilde{\A}\|^{1/2} \cdot \epsilon_{2,2} \big\|(\tilde{\A}^\top\tilde{\A}+\tilde{\lambda}\I)^{1/2}\v^*\big\|
\end{align*}
where in the fourth step we use the commutable property, i.e., $(\tilde{\A}\tilde{\A}^\top+\tilde{\lambda}\I)^{1/2}\tilde{\A} = \tilde{\A} (\tilde{\A}^\top\tilde{\A}+\tilde{\lambda}\I)^{1/2}$. Notice that from the definition of $\u^*$ and $\v^*$ we naturally have $(\tilde{\A}^\top\tilde{\A})\u^* = \r =(\tilde{\A}^\top\tilde{\A}+\tilde{\lambda}\I)\v^*$, which gives
$\u^* = \frac{1}{\tilde{\lambda}} (\tilde{\A}^\top\tilde{\A} + \tilde{\lambda}\I)(\u^*-\v^*) = (\tilde{\A}^\top\tilde{\A} + \tilde{\lambda}\I)\z^*$, and also $\v^* = \frac{1}{\tilde{\lambda}}\tilde{\A}^\top\tilde{\A}(\u^*-\v^*) = \tilde{\A}^\top\tilde{\A}\z^*$. By using these facts we can express $\u^*$ and $\v^*$ in terms of $\z^*$, and have
\begin{align*}
& ~ \tilde{\lambda} \|\hat{\z} - \z^*\|_{\W^2+\tilde{\lambda}\W}\\
\leq & ~ \epsilon_{2,1} \|\tilde{\A}\tilde{\A}^\top+\tilde{\lambda}\I\|^{1/2}\cdot \|\tilde{\A}\u^*\| + \epsilon_{2,2} \|\tilde{\A}^\top\tilde{\A}\|^{1/2} \cdot \| (\tilde{\A}^\top\tilde{\A} + \tilde{\lambda}\I)^{1/2} \v^*\| \\
= & ~ \epsilon_{2,1} \|\tilde{\A}\tilde{\A}^\top+\tilde{\lambda}\I\|^{1/2}\cdot \|\tilde{\A}(\tilde{\A}^\top\tilde{\A} + \tilde{\lambda}\I)\z^*\| + \epsilon_{2,2} \|\tilde{\A}^\top\tilde{\A}\|^{1/2} \cdot \| (\tilde{\A}^\top\tilde{\A} + \tilde{\lambda}\I)^{1/2} \tilde{\A}^\top\tilde{\A}\z^*\| \\
= & ~ \epsilon_{2,1} \|\tilde{\A}\tilde{\A}^\top+\tilde{\lambda}\I\|^{1/2}\cdot \|(\tilde{\A}\tilde{\A}^\top+\tilde{\lambda}\I) \tilde{\A}\z^*\| + \epsilon_{2,2} \|\tilde{\A}^\top\tilde{\A}\|^{1/2} \cdot \|\tilde{\A}^\top(\tilde{\A}\tilde{\A}^\top+\tilde{\lambda}\I)^{1/2}\tilde{\A}\z^*\| \\
\leq & ~ \epsilon_{2,1} \|\tilde{\A}\tilde{\A}^\top+\tilde{\lambda}\I\|\cdot \|(\tilde{\A}\tilde{\A}^\top+\tilde{\lambda}\I)^{1/2} \tilde{\A}\z^*\| + \epsilon_{2,2} \|\tilde{\A}\tilde{\A}^\top\| \cdot \|(\tilde{\A}\tilde{\A}^\top+\tilde{\lambda}\I)^{1/2}\tilde{\A}\z^*\| \\
= & ~ \left(\epsilon_{2,1} \|\tilde{\A}\tilde{\A}^\top+\tilde{\lambda}\I\| + \epsilon_{2,2} \|\tilde{\A}\tilde{\A}^\top\| \right)\cdot \|\z^*\|_{\W^2+\tilde{\lambda}\W}.
\end{align*}
Thus given $\epsilon_0>0$, by choosing $\epsilon_{2,1} \leq \frac{\epsilon_0\tilde{\lambda}}{2\|\tilde{\A}\tilde{\A}^\top+\tilde{\lambda}\I\|}$ and $\epsilon_{2,2} \leq \frac{\epsilon_0\tilde{\lambda}}{2\|\tilde{\A}\tilde{\A}^\top\|}$, we have
\begin{align}\label{eq:level_2_req}
\|\hat{\z} - \z^*\|_{\W^2+\tilde{\lambda}\W} \leq \epsilon_0 \cdot \|\z^*\|_{\W^2+\tilde{\lambda}\W}.
\end{align}
Notice that $\z^* = \M_2^{-1}\r$. With this result, we apply Lemma~\ref{lem:stable_lanczos} to linear system $(\C^\top\C+\tilde{\lambda}\W)\y = \C^\top\r$ with preconditioner $\M_2$, and obtain $\hat{\y}$ such that 
\begin{align*}
\|\hat{\y} - \y^*\|_{\C^\top\C+\tilde{\lambda}\W} \leq \epsilon_1 \cdot \|\y^*\|_{\C^\top\C+\tilde{\lambda}\W}
\end{align*}
in $O(\sqrt{\kappa_{\M_2}} \log (\kappa_{\M_2} / \epsilon_1)) = O(\log 1/\epsilon_1)$ iterations, since $\kappa_{\M_2} = O(1)$. Notice that this is exactly the requirement we need for Level-1 (see Eq.\eqref{eq:level_1_req}). Moreover, we have the following bounds:
\begin{align*}
\frac{\epsilon_0\tilde{\lambda}}{2\|\tilde{\A}\tilde{\A}^\top+\tilde{\lambda}\I\|} \geq & ~ \frac{\epsilon_0\tilde{\lambda}}{4\|\A\A^\top+\tilde{\lambda}\I\|} \geq \frac{\epsilon_0\lambda_0}{4(\sigma_1^2 + \lambda_0)} = \frac{\epsilon_0\lambda_0/\sigma_n^2}{4\kappa^2 + 4\lambda_0/\sigma_n^2} \\
= & ~ \frac{\epsilon_0(n-l)\bar{\kappa}_{l}^2}{2 \kappa^2 l + 4(n-l)\bar{\kappa}_{l}^2} \geq \frac{\epsilon_0(n-l)}{2\kappa^2 l + 4(n-l)} \geq \frac{\epsilon_0}{2\kappa^2 l + 4} \geq \frac{\epsilon_0}{4\kappa^2 l}
\end{align*}
where we use the fact that $\frac{\lambda_0}{\sigma_n^2} = \frac{2(n-l)}{l}\bar{\kappa}_{l}^2$ and $\bar{\kappa}_{l} \geq 1$, and similarly,
\begin{align*}
\frac{\epsilon_0\tilde{\lambda}}{2\|\tilde{\A}\tilde{\A}^\top\|} \geq \frac{\epsilon_0\tilde{\lambda}}{4\|\A\A^\top\|} \geq \frac{\epsilon_0\lambda_0}{4\|\A\A^\top\|} = \frac{\epsilon_0\lambda_0}{4\sigma_1^2} \geq \frac{\epsilon_0(n-l)\bar{\kappa}_{l,\lambda}^2}{2\kappa^2 l} \geq \frac{\epsilon_0}{4\kappa^2 l}.
\end{align*}
Thus we can choose $\epsilon_{2,1} = \epsilon_{2,2} = \epsilon_0/4\kappa^2 l$ to guarantee that the requirement we need for Level-2 (see Eq.\eqref{eq:level_2_req}) holds. Finally we step into the innermost Level-3.

\paragraph{Part 4: Level-3 Preconditioning.}
As the innermost level (Algorithm~\ref{alg:msp_rec_func_inner}), we use preconditioned Lanczos to solve two positive definite linear systems $(\tilde{\A}^\top\tilde{\A})\u = \r$ and $(\tilde{\A}^\top\tilde{\A} + \tilde{\lambda}\I)\v = \r$, with preconditioners $\M_{3a} \coloneqq \hat{\A}^\top\hat{\A}$ and $\M_{3b} \coloneqq \hat{\A}^\top\hat{\A} + \tilde{\lambda}\I$ respectively, where $\hat{\A} \coloneqq \mathbf{\Phi}\tilde{\A} \in\R^{\phi\times s}$ and $\mathbf{\Phi}$ is the oblivious embedding matrix defined by Lemma~\ref{lem:precondition_sparse} with $\phi = O(s + \log (1/\delta))$. With probability $1 - \delta / 3$ we have $\kappa_{\M_{3a}} = \kappa_{\M_{3b}} = O(1)$. Since in Algorithm~\ref{alg:msp_rec_main} we pre-compute $\M_{3a}^{-1}$ and $\M_{3b}^{-1}$ exactly, the error terms for these two preconditioners are $0$. By using Lemma~\ref{lem:stable_lanczos}, we need $O(\sqrt{\kappa_{\M_{3a}}} \log (\kappa_{\M_{3a}}/\epsilon_{2,1}))$ and $O(\sqrt{\kappa_{\M_{3b}}} \log (\kappa_{\M_{3b}} / \epsilon_{2,2}))$ iterations respectively to guarantee
\begin{align*}
\begin{cases}
\|\hat{\u} - \u^*\|_{\tilde{\A}^\top\tilde{\A}} \leq \epsilon_{2,1} \cdot \|\u^*\|_{\tilde{\A}^\top\tilde{\A}}; \\
\|\hat{\v} - \v^*\|_{\tilde{\A}^\top\tilde{\A}+\tilde{\lambda}\I} \leq \epsilon_{2,2} \cdot \|\v^*\|_{\tilde{\A}^\top\tilde{\A}+\tilde{\lambda}\I}.
\end{cases}
\end{align*}
Furthermore, since we choose $\epsilon_{2,1} = \epsilon_{2,2} = \epsilon_0 / 4\kappa^2 l$, thus the number of iterations in Level-3 are $O(\sqrt{\kappa_{\M_{3a}}} \log (\kappa_{\M_{3a}}/\epsilon_{2,1})) = O(\log 1/\epsilon_{2,1}) = O(\log(\kappa l / \epsilon_0))$. Notice that the above guarantees are exactly the requirements we need for Level-2, thus, by going back from Level-3 $\to$ Level-2 $\to$ Level-1 and applying a union bound over the probabilities of ``$\kappa_{\M} = O(\bar{\kappa}_{l,\lambda}^2 n/l), \kappa_{\M_2} = \kappa_{\M_{3a}} = \kappa_{\M_{3b}} = O(1)$'', we conclude that $\|\tilde{\x} - \x^*\|_{\A^\top\A+\lambda\I} \leq \epsilon \|\x^*\|_{\A^\top\A+\lambda\I}$ holds with probability $1-\delta$.

\paragraph{Part 5: Total Cost.}
Finally we consider the total cost of Algorithm~\ref{alg:msp_rec_main} together with Algorithm~\ref{alg:msp_rec_func_outer} and \ref{alg:msp_rec_func_inner}. We first pre-compute $\tilde{\A} = \A\S^\top$ which takes $O(\nnz(\A) \log(l/\delta))$, and also compute $\hat{\A} = \mathbf{\Phi}\tilde{\A}$ which takes $O(\nnz(\tilde{\A}) \log(s/\delta) + s^2\log^4(s/\delta))$ according to Lemma~\ref{lem:precondition_sparse}. With these, we can compute the Level-3 preconditioners $\M_{3a}^{-1} = (\hat{\A}^\top\hat{\A})^{-1}$ and $\M_{3b}^{-1}=(\hat{\A}^\top\hat{\A}+\tilde{\lambda}\I)^{-1}$ in time $O(s^{\omega})$. To conclude, line 2-6 of Algorithm~\ref{alg:msp_rec_main} takes in total:
\begin{align*}
O(\nnz(\tilde{\A}) \log(s/\delta) + s^2\log^4(s/\delta) + s^{\omega}) = O(\nnz(\A)\log^2(l/\delta) + l^{\omega}\log^\omega(l/\delta)).
\end{align*}
Within each iteration, we need to compute matrix-vector products of matrix $\A^\top\A+\lambda\I$ which takes $O(\nnz(\A))$, and call $\mathsf{SolveM1}$ as in Algorithm~\ref{alg:msp_rec_func_outer}. Denote $T_{\M^{-1}}$ as the time of calling $\mathsf{SolveM1}$, and denote $T_{\M_2^{-1}}$ as the time of calling $\mathsf{SolveM2}$ as in Algorithm~\ref{alg:msp_rec_func_inner}. Observe that in each iteration of Algorithm~\ref{alg:msp_rec_func_inner} we need to compute $\tilde{\A}^\top\tilde{\A}\r = \S\A^\top\A\S^\top\r$ for any given vector $\r$, which takes $O(n\log(l/\delta) + \nnz(\A)) = O(\nnz(\A)\log(l/\delta))$. Since there are $O(\log (\kappa l/\epsilon_0))$ iterations, the complexity for Algorithm~\ref{alg:msp_rec_func_inner} becomes
\begin{align*}
T_{\M_2^{-1}} = O\left(\nnz(\A)\log(l/\delta) \log(\kappa l/\epsilon_0) \right).
\end{align*}
In each iteration of Algorithm~\ref{alg:msp_rec_func_outer} we need to compute $\tilde{\A}^\top\tilde{\A}\r$ which is similar, thus we have
\begin{align*}
T_{\M^{-1}} = & ~ O((\nnz(\A)\log(l/\delta) +T_{\M_2^{-1}})\cdot\log 1/\epsilon_1) \\
= & ~ O(\nnz(\A)\log(l/\delta)\log(\kappa l/\epsilon_0)\log 1/\epsilon_1) \\
= & ~ O(\nnz(\A)\log(l/\delta)\log(\kappa l/\epsilon_0) \log (\kappa / \epsilon_0)).
\end{align*}
Finally, since there are $t_{\max} = O(\bar{\kappa}_{l,\lambda}\sqrt{n/l}\cdot \log (\bar{\kappa}_{l,\lambda} n/\epsilon))$ iterations in total, we compute the overall complexity of Algorithm~\ref{alg:msp_rec_main} as the following:
\begin{align*}
& ~ O(\underbrace{\nnz(\A)\log^2(l/\delta) + l^{\omega}\log^\omega(l/\delta)}_{\text{preconditioning}} + \underbrace{(\nnz(\A)+T_{\M^{-1}})\cdot \bar{\kappa}_{l,\lambda}\sqrt{n/l} \log (\bar{\kappa}_{l,\lambda} n/\epsilon)}_{\text{iteration}}) \\
= & ~ O(\nnz(\A)\log^2(l/\delta) + l^{\omega}\log^\omega(l/\delta) + \nnz(\A)\log(l/\delta)\log(\kappa l/\epsilon_0) \log(\kappa/\epsilon_0)\cdot \bar{\kappa}_{l,\lambda}\sqrt{n/l}\log (\bar{\kappa}_{l,\lambda} n/\epsilon))\\
= & ~ O(\nnz(\A)\log^2(l/\delta)\log^2(\kappa/\epsilon_0)\cdot \bar{\kappa}_{l,\lambda}\sqrt{n/l}\log (\bar{\kappa}_{l,\lambda} n/\epsilon) + l^{\omega} \log^\omega (l/\delta)) \\
= & ~ \tilde{O}\left(\nnz(\A)\sqrt{\frac{n}{l}} \bar{\kappa}_{l,\lambda}\log^3 (\kappa/\epsilon) + l^{\omega}\right).
\end{align*}

\subsection{Proof of Lemma \ref{lem:spec_approx}}
For $\tilde{\lambda} = \lambda+\lambda_0\geq \frac{2}{l}\sum_{i>l}\sigma_i^2$, denote $\Sig_{\lambda} = \A\A^\top+\tilde{\lambda}\I$ and $\B = \A^\top\Sig_{\tilde{\lambda}}^{-1/2}$. Here we list some basic properties of $\B$:
\begin{align*}
\begin{cases}
\|\B\|^2 = \|\B\B^\top\| = \|\A^\top\Sig_{\tilde{\lambda}}^{-1}\A\| \leq 1 \\
\|\B\|_F^2 = \tr(\B\B^\top) = \tr(\A\A^\top(\A\A^\top+\tilde{\lambda}\I)^{-1}) =: d_{\tilde{\lambda}}
\end{cases}
\end{align*}
By using Lemma~\ref{lem:cohen_bound}, we know that the sparse subspace embedding matrix with $s= O(d_{\tilde{\lambda}} \log(d_{\tilde{\lambda}}/\delta))$ and each column having $\gamma = O(\log(d_{\tilde{\lambda}}/\delta))$ satisfies the $(\frac{1}{6}, \delta, d_{\tilde{\lambda}}, \log(d_{\tilde{\lambda}}/\delta))$-OSE moment property, thus by using Lemma~\ref{lem:OSE_req}, with probability $1-\delta$ we have
\begin{align*}
\|\B^\top\S^\top\S\B - \B^\top\B\| \leq \frac{1}{6} (\|\B\|^2 + 2\|\B\|_F^2/d_{\tilde{\lambda}}) \leq \frac{1}{2}.
\end{align*}
Conditioned on this guarantee, we have the following holds:
\begin{align*}
& ~ \B^\top\B - \frac{1}{2}\I \preceq \B^\top\S^\top\S\B \preceq \B^\top\B+\frac{1}{2}\I \\
\Leftrightarrow & ~ \Sig_{\tilde{\lambda}}^{-1/2}\A\A^\top\Sig_{\tilde{\lambda}}^{-1/2} - \frac{1}{2}\I \preceq \Sig_{\tilde{\lambda}}^{-1/2}\A\S^\top\S\A^\top \Sig_{\tilde{\lambda}}^{-1/2} \preceq \Sig_{\tilde{\lambda}}^{-1/2}\A\A^\top\Sig_{\tilde{\lambda}}^{-1/2}+\frac{1}{2}\I \\
\Leftrightarrow & ~ \A\A^\top - \frac{1}{2}(\A\A^\top+\tilde{\lambda}\I) \preceq \A\S^\top\S\A^\top \preceq \A\A^\top+\frac{1}{2}(\A\A^\top+\tilde{\lambda}\I) \\
\Leftrightarrow & ~ \frac{1}{2} (\A\A^\top+\tilde{\lambda}\I) \preceq \A\S^\top\S\A^\top+\tilde{\lambda}\I \preceq \frac{3}{2}(\A\A^\top+\tilde{\lambda}\I) \\
\Rightarrow & ~ \tilde{\A}\tilde{\A}^\top + \tilde{\lambda}\I \approx_2 \A\A^\top+\tilde{\lambda}\I.
\end{align*}
Finally we bound $d_{\tilde{\lambda}}$ by $l$ by using the fact that $\tilde{\lambda} \geq \frac{2}{l}\sum_{i>l}\sigma_i^2$:
\begin{align*}
d_{\tilde{\lambda}} = \sum_{i=1}^n \frac{\sigma_i^2}{\sigma_i^2 + \tilde{\lambda}} = \sum_{i=1}^l \frac{\sigma_i^2}{\sigma_i^2 + \tilde{\lambda}} + \sum_{i=l+1}^n \frac{\sigma_i^2}{\sigma_i^2 + \tilde{\lambda}} \leq l+\sum_{i=l+1}^n \frac{\sigma_i^2 l}{2\sum_{j>l}\sigma_j^2} = \frac{3}{2}l,
\end{align*}
thus we conclude that $s= O(l \log(l/\delta))$ and $\gamma=O(\log(l/\delta))$ suffices for the spectral approximation.

\section{Analysis of Inexact Preconditioned Lanczos Iteration}\label{s:lanczos}
Our \AlgFullName{} (\AlgName{}) algorithms described in \Cref{s:psd-proof,sec:proof_rec} rely on multiple calls to a preconditioned Lanczos algorithm. 
This method requires routines for matrix-vector multiplication with $\A$, as well as matrix-vector multiplication with the inverse of the preconditioner, $\M$. A key feature of our framework is that, in many cases, we do not explicitly construct $\M^{-1}$, but rather compute $\M^{-1}\r$ for a given vector $\r\in \R^{n}$ \emph{approximately}. It is not clear a priori that the outer iterative method used in our algorithms is robust to this approximation, i.e. that Lanczos still converges as quickly as the case when $\M^{-1}$ is applied exactly. In particular, while there has been significant work showing that the \emph{unpreconditioned} Lanczos method is robust to inexact matrix-vector multiplications, and actually to round-off error in all arithmetic operations \cite{paige1976error,greenbaum_89,musco_musco_sidford_18}, there is less work on analyzing \emph{preconditioned} Lanczos. 

One line of work studies the closely related Preconditioned Conjugate Gradient (PCG) method, which returns exactly the same output as Lanczos when all computations are performed exactly \cite{golub1999inexact}. That work establishes that PCG still converges when $\M^{-1}$ is applied approximately. However, it has the disadvantage of only providing a \emph{non-accelerated} convergence rate. I.e., it establishes convergence depending on the condition number $\kappa_\M = \kappa(\M^{-1/2}\A\M^{-1/2})$ of the preconditioned system. On the other hand, if $\M^{-1}$ is applied exactly, a dependence on $\sqrt{\kappa_\M}$ is achievable. This better $\sqrt{\kappa_\M}$ dependence is critical for obtaining our final results. Other efforts to analyze the stability of iterative methods also fail to obtain an accelerated rate. For example, there has been interesting recent work on the stability of preconditioned gradient descent \cite{epperly23}, but any gradient descent type method will have convergence depending on $\kappa_\M$ instead of $\sqrt{\kappa_\M}$.

We address this issue by proving a new result for the Lanczos method with inexact preconditioning that maintains the accelerated rate. We suspect that a similar result can be achieved for PCG, which is more commonly used in practice than Lanczos, although we leave this effort to future work. As discussed in \Cref{s:techniques}, we also note that our analysis of Lanczos is not strictly necessary for this paper. For example, \cite{SpielmanTeng:2014} give an analysis of inexact preconditioned \emph{Chebyshev iteration}, which could have been used instead. However, the preconditioned Lanczos method is easier to implement (e.g., does not require eigenvalue bounds) and much more widely used in practice since it tends to converge much more quickly than Chebyshev iteration, 
even though it has the same worst-case guarantees \cite{lanczos_near_opt}. As such, we expect that our stability analysis of the preconditioned Lanczos method will be of general interest beyond this work.

\subsection{Inexact Preconditioned Lanczos Method}
Concretely, in this section we analyze the preconditioned Lanczos method described in \Cref{alg:lanczos_left_precon}.
The specific implementation of Lanczos in \Cref{alg:lanczos_left_precon} is based on an implementation of an unpreconditioned Lanczos method whose numerical stability was analyzed in seminal work by Christopher Paige \cite{paigeThesis,paige1976error}. Paige's work was later used to show that Lanczos can stably solve linear systems and more general matrix function problems on finite precision computers \cite{greenbaum_89,musco_musco_sidford_18}. Our analysis will rely on these previous results. 
Formally, we prove the following result:
\begin{theorem}
\label{thm:stable_lanczos}
Consider solving $\A\x=\b$ for positive definite $\A$ using \Cref{alg:lanczos_left_precon} provided with a function $\mathsf{SolveM}$ that, for some positive semidefinite preconditioner $\M$ and any vector $\r$ returns
\begin{align}
\label{eq:m_assump_informal}
\|\mathsf{SolveM}(\r) - \M^{-1}\r\|_{\M} \leq \epsilon_0 \cdot \|\M^{-1}\r\|_{\M}.\end{align}
If $\epsilon_0 \leq (\frac{\epsilon}{\kappa_{\M} n})^c$ for a fixed constant $c>0$, where ${\kappa_{\M}}$ is the condition number of $\M^{-1/2}\A\M^{-1/2}$, then the algorithm outputs $\x_t$ such that $\|\x_t - \A^{-1}\b\|_{\A} \leq \epsilon\|\A^{-1}\b\|_{\A}$ in $t= O(\sqrt{\kappa_{\M}} \log (\kappa_{\M}/\epsilon))$ iterations.
\end{theorem}

\begin{algorithm}[!ht]
\caption{Preconditioned Lanczos Iteration $\textsc{}$}
\label{alg:lanczos_left_precon}
\begin{algorithmic}[1]
\State \textbf{Input: }positive definite $\A\in\R^{n\times n}$, vector $\b \in \R^n$, function $\mathsf{SolveM}$, \# of iterations $t$;
\State \textbf{Output: }vector $\x_t\in\R^n$ that approximates $\A^{-1}\b$;
\State $\overline{\w_0} \gets \mathsf{SolveM}(\b), z' = \sqrt{\langle \b, \overline{\w_0}\rangle};$
\State $\underline{\q_0} = \mathbf{0}, \overline{\q_1} = \w' / z, \underline{\q_1} = \b / z, \beta_1' = 0$;
\For{$i = 1, \ldots, t$}
\State $\underline{\u_{i}} \leftarrow \A\overline{\q_i} - \beta_i' \underline{\q_{i-1}}$;
\State $\alpha_i' \leftarrow \langle \underline{\u_{i}}, \overline{\q_i}\rangle$;
\State $\underline{\w_{i}} \leftarrow \underline{\u_{i}} - \alpha_i' \underline{\q_{i}}$,  $\overline{\w_i} = \mathsf{SolveM}(\underline{\w_{i}})$;
\State $\beta_{i+1}' \leftarrow \langle \underline{\w_{i}},\overline{\w_i}\rangle^{1/2}$;
\State \textbf{if} $\beta_{i+1}' == 0$ \textbf{then}
\State \quad \textbf{break};
\State \textbf{end if}
\State $\underline{\q_{i+1}} \leftarrow \underline{\w_{i}} / \beta_{i+1}'$, $\overline{\q_{i+1}} \leftarrow \overline{\w_{i}} / \beta_{i+1}'$;  
\EndFor
\State $\T'\in \R^{t\times t} \leftarrow \left[\begin{array}{cccc}
\alpha_1' & \beta_2' & & 0 \\
\beta_2' & \alpha_2' & \ddots & \\
& \ddots & \ddots & \beta_t' \\
0 & & \beta_t' & \alpha_t'
\end{array}\right], \quad \overline{\mathbf{Q}}\in \R^{n\times t} \leftarrow\left[\begin{array}{lll}
\overline{\mathbf{q}_1} & \ldots & \overline{\mathbf{q}_t}
\end{array}\right]$; \\
\Return $\x_t' = z' \overline{\Q} \T'^{-1} \e_1$ where $\e_1\in \R^t$ denotes the first standard basis vector.
\end{algorithmic}
\end{algorithm}

\begin{algorithm}[!ht]
\caption{Symmetric Preconditioned Lanczos Iteration.}
\label{alg:lanczos_sym_precon}
\begin{algorithmic}[1]
\State \textbf{Input: }positive definite $\A\in\R^{n\times n}$, vector $\b \in \R^n$, positive definite  preconditioner $\M\in\R^{n\times n}$, \# of iterations $t$;\State \textbf{Output: }vector $\y\in\R^n$ that approximates $\A^{-1}\b$;
\State $\w_0 \gets \M^{-1/2}\b$, $z \gets \|\w_0\|$;
\State $\q_0 \gets \mathbf{0}, \q_1 \gets \w / z, \beta_1 = 0$;
\For{$i = 1, \ldots, t$}
\State $\u_{i} \leftarrow \M^{-1/2}\A\M^{-1/2}\q_i - \beta_i \q_{i-1}$;
\State $\alpha_i \leftarrow \langle \u_{i}, \q_i\rangle$;
\State $\w_{i} \leftarrow \u_{i} - \alpha_i \q_{i}$;
\State $\beta_{i+1} \leftarrow \|\w_{i}\|$;
\State \textbf{if} $\beta_{i+1} == 0$ \textbf{then}
\State \quad \textbf{break};
\State \textbf{end if}
\State $\q_{i+1}\leftarrow \w_{i} / \beta_{i+1}$; 
\EndFor
\State $\T\in \R^{t\times t} \leftarrow \left[\begin{array}{cccc}
\alpha_1 & \beta_2 & & 0 \\
\beta_2 & \alpha_2 & \ddots & \\
& \ddots & \ddots & \beta_k \\
0 & & \beta_k & \alpha_k
\end{array}\right], \quad \mathbf{Q}\in \R^{n\times t} \leftarrow\left[\begin{array}{lll}
\mathbf{q}_1 & \ldots & \mathbf{q}_k
\end{array}\right]$; \\
\Return $\x_t = \M^{-1/2}(z\Q \T^{-1} \e_1)$, where $\e_1\in \R^t$ denotes the first standard basis vector.
\end{algorithmic}
\end{algorithm}

We prove Theorem \ref{thm:stable_lanczos} in an indirect way by comparing the output of \Cref{alg:lanczos_left_precon} to the output of a more easily analyzed variant of Lanczos with \emph{symmetric preconditioning}, which is given as \Cref{alg:lanczos_sym_precon}. This variant assumes access to the symmetrically preconditioned matrix $\M^{-1/2}\A\M^{-1/2}$. To motivate the algorithm, observe that:
\begin{align*}
    \A^{-1}\b = \M^{-1/2}(\M^{-1/2}\A\M^{-1/2})^{-1}\M^{-1/2}\b.
\end{align*}
\Cref{alg:lanczos_sym_precon} is obtained by just running the standard unpreconditioned Lanczos method (see e.g., \cite{musco_musco_sidford_18}) with right hand side $\M^{-1/2}\b$ and matrix $(\M^{-1/2}\A\M^{-1/2})$, and then multiplying the final result by $\M^{-1/2}$.
In our application, we only have access to $\M^{-1}$, so we cannot implement \Cref{alg:lanczos_sym_precon} directly: it is only used for our analysis of \Cref{alg:lanczos_left_precon}. In particular, we will take advantage of the fact that it can be easily shown that \Cref{alg:lanczos_sym_precon} is robust to implementation in \emph{finite precision arithmetic}. This is because, unlike  \Cref{alg:lanczos_left_precon}, \Cref{alg:lanczos_sym_precon}  is formally equivalent to {unpreconditioned} Lanczos run with a particular choice of inputs. The robustness of the \emph{unpreconditioned} Lanczos method to round-off errors has been widely studied.  

Formally, suppose all basic arithmetic operators are computed up to \emph{relative accuracy}. I.e., for any operation $x\circ y$ where $\circ \in \{+,-,\times,\div\}$, the computer running \Cref{alg:lanczos_sym_precon} returns a result 
\begin{align}
\label{eq:finite_precision}
    z &= (1+\delta)(x\circ y) & &\text{such that} & |\delta|&\leq \mach,
\end{align}
where $\mach$ is the machine precision. Similarly, for any $x$, the computer can compute $z = (1+\delta)(x\circ y)$  where $|\delta| \leq \mach$. In this setting, we have the following bound:
\begin{fact}[Corollary of Theorem 1 \cite{musco_musco_sidford_18}]
\label{f:standard_system_stable}
Suppose \Cref{alg:lanczos_sym_precon} is given access to the exact value of $\M^{-1/2}\A\M^{-1/2}$ and $\M^{-1/2}\b$ and is then implemented on a finite precision computer with  machine precision $\mach = \poly(\epsilon/n\kappa_{\M})$, except that the final multiplication by $\M^{-1/2}$ on line 16 is performed exactly. Here $\kappa_\M = \kappa(\M^{-1/2}\A\M^{-1/2})$. Then the algorithm will return $\x_t$ such that $\|\x_t - \A^{-1}\b\|_{\A} \leq \epsilon\|\A^{-1}\b\|_{\A}$ after $t= O\left(\sqrt{\kappa_{\M}} \log (\kappa_{\M}/\epsilon)\right)$ iterations.
\end{fact}
\begin{proof}
This fact is obtained by applying  Theorem 1 of \cite{musco_musco_sidford_18} to solving the positive definite preconditioned system $\M^{-1/2}\A\M^{-1/2}\x = \M^{-1/2}\b$, which has solution $\M^{1/2}\A^{-1}\b$. Observe that \Cref{alg:lanczos_sym_precon} is exactly equivalent to the Lanczos method analyzed in that work, except for the last step: to return an approximate solution to $\M^{-1/2}\A\M^{-1/2}\x = \M^{-1/2}\b$, we should return just 
\begin{align*}
    \y = z\Q \T^{-1} \e_1
\end{align*} instead of $\x_t = \M^{-1/2}(z\Q \T^{-1} \e_1)$. 
Let $\lambda_1 \geq \ldots \geq \lambda_n > 0$ denote the eigenvalues of the $\M^{-1/2}\A\M^{-1/2}$. We apply Theorem 1 with $f(x) = 1/x$ and $\eta = \lambda_n/2$, and use the well known fact (see, e.g. \cite{Shewchuk:1994}) that for any $\delta < 1$,  there exists a degree $O(\sqrt{\kappa_{\M}}\log(1/\delta))$ polynomial such that:
\begin{align*}
\max_{x\in [\lambda_n, \lambda_1]} |p(x) - 1/x| \leq \frac{\delta}{\lambda_n}.
\end{align*}
Setting $\delta = \poly(\epsilon/\kappa_{\M})$ and $\mach = \poly(n\kappa_{\M}/\epsilon)$, Theorem 1 from \cite{musco_musco_sidford_18} implies that after $ t= O(\sqrt{\kappa_{\M}} \log(\kappa_{\M}/\epsilon))$ iterations,
\begin{align}
\label{eq:direct_thm1_bound}
\|\y - \M^{1/2}\A^{-1}\b\| \leq \frac{\epsilon}{\kappa_{\M} \lambda_n} \|\M^{-1/2}\b\|.
\end{align}
Using that $\x_t = \M^{-1/2}\y$ exactly (since we assumed the last multiplication by $\M^{-1/2}$ is perfomed exactly on Line 16, we have:
\begin{align*}
\|\y - \M^{1/2}\A^{-1}\b\|  &= \|(\A^{1/2}\M^{-1/2})^{-1}\A^{1/2}\M^{-1/2}\y - (\A^{1/2}M^{-1/2})^{-1}\A^{-1/2}\b\| \\
&\geq \frac{1}{\sqrt{\lambda_1}} \|\A^{1/2}\x_t - \A^{-1/2}\b\| = \frac{1}{\sqrt{\lambda_1}}  \|\x_t - \A^{-1}\b\|_{\A}.
\end{align*}
We further have that:
\begin{align*}
    \frac{\epsilon}{\kappa_{\M} \lambda_n} \|\M^{-1/2}\b\| &= \frac{\epsilon}{\kappa_{\M} \lambda_n} \|(\A^{-1/2}\M^{1/2})^{-1}\A^{-1/2}\b\|\leq \frac{\epsilon}{\kappa_{\M} \lambda_n}\sqrt{\lambda_1}\|\A^{-1}\b\|_{\A}. 
\end{align*}
Plugging both of these bounds into \eqref{eq:direct_thm1_bound}, we conclude that, as desired. 
\begin{align*}
\|\x_t - \A^{-1}\b\|_{\A} &\leq \frac{\epsilon{\lambda_1}}{\kappa_{\M} \lambda_n}\|\A^{-1}\b\|_{\A} = \epsilon\|\A^{-1}\b\|_{\A}.
\end{align*}
\end{proof}
We can actually slightly strengthen \Cref{f:standard_system_stable}. In particular, the analysis in \cite{musco_musco_sidford_18} is based in a black-box way on a bound from \cite{paige1976error} on the output of the  general Lanczos tridiagonalization method in finite precision. This result is stated as Theorem 8 in \cite{musco_musco_sidford_18} and is the first theorem in \cite{paige1976error}. So, the fact actually holds for \emph{any sequence of round-off errors} for which this critical theorem remains true.
\begin{fact}
    \label{f:more_general} The result of \Cref{f:standard_system_stable} holds when \Cref{alg:lanczos_sym_precon} is implemented with any sequence of round-off errors for which the main Theorem of  \cite{paige1976error} (Theorem 8 in \cite{musco_musco_sidford_18}) holds with $\mach = \poly(n\kappa_{\M}/\epsilon)$. This includes, e.g., any sequence of adversarially chosen round-off errors on basic arithmetic operations that satisfy \eqref{eq:finite_precision}, even if those are not precisely the errors that would be made in a standard implementation of finite precision arithmetic.
\end{fact}
We leverage \Cref{f:more_general} to prove \Cref{thm:stable_lanczos} by arguing that the output of \Cref{alg:lanczos_left_precon}, the preconditioned Lanczos method that we hope to analyze, is \emph{exactly equivalent} to the output of \Cref{alg:lanczos_sym_precon} for some choice of round-off errors that satisfy the requirements to prove Paige's theorem. The proof is completed below. It is not self-contained: the reader will require access to \cite{paige1976error}. As of April 2024, a copy can be found on the \href{https://www.cs.mcgill.ca/~chris/pubClassic/76JIMA001.pdf}{author's webpage}.

\begin{proof}[Proof of \Cref{thm:stable_lanczos}]
Before beginning the proof, we make a remark on notation. The variable names in \Cref{alg:lanczos_left_precon} are chosen to mirror the variable names in \Cref{alg:lanczos_sym_precon}. In particular, if \Cref{alg:lanczos_left_precon} and  \Cref{alg:lanczos_sym_precon} are implemented in exact arithmetic and with exact applications of all matrix-vector multiplications, it can be checked that any variable $\x$ satisfies the relationship
\begin{align*}
\overline{\x} &= \M^{-1/2}\x & &\text{and} & \underline{\x} &= \M^{1/2}\x.
\end{align*}
E.g., at Line 6 of \Cref{alg:lanczos_left_precon}, $\underline{\u_i}$ would exactly equal $\M^{1/2}\u_i$ computed at Line 6 of \Cref{alg:lanczos_sym_precon}. Additionally, any variable that appears with a single tick mark in \Cref{alg:lanczos_left_precon} would be exactly equal to the corresponding variable in \Cref{alg:lanczos_sym_precon}. I.e. we would have that $z' = z$ on Line 3.

To motivate \Cref{alg:lanczos_left_precon}, note that it is essentially equivalent to \Cref{alg:lanczos_sym_precon}, except that, instead of keeping track of $\Q = [\q_1, \ldots, \q_k]$, it keeps track of $\M^{-1/2}\Q = [\overline{\q_1}, \ldots, \overline{\q_k}]$. This can be done without every having access to $\M^{-1/2}$ and, importantly, allows use to avoid the final step of \Cref{alg:lanczos_sym_precon}, which requires multiplying $\M^{-1/2}$ by a vector in the span of $\Q$.

Concretely, we will prove that there is a sequence of round-off errors in \Cref{alg:lanczos_sym_precon} for which Paige's main theorem holds, and for which  $\T' = \T$ and $\overline{\Q} = \M^{-1/2}\Q$, where $\T'$ and $\overline{\Q}$ are the quantities computed by \Cref{alg:lanczos_left_precon} run with a procedure $\mathsf{SolveM}$ for applying $\M^{-1}$ that satisfies the accuracy guarantee of \eqref{eq:m_assump_informal}. We will actually argue that this equivalence holds when \Cref{alg:lanczos_sym_precon} is run on a slight perturbation of $\b$, $\b'$. We will account for this difference towards the end of the proof. 

Our proof will proceed via induction. In particular, we will prove that for all $i$,
\begin{align}
\label{eq:3}
\alpha_i'&=\alpha_i & \beta_i' &= \beta_i, &
\overline{\q_i} &= \M^{-1/2}\q_i & &\text{and} & \|\underline{\q_{i}} -  \M^{1/2}{\q_{i}}\|_{\M^{-1}}  \leq \mach.
\end{align}
To establish this bound, we will inductively assume that it holds for all $j\leq i$. We begin with the base cases, which includes any variable set outside of the for loop on Line 5.

\medskip\noindent\textbf{Base Case.} We start by noting that trivially, $\mathbf{0} = \underline{\q_0} = \M^{1/2}\q_0 = \mathbf{0}$ and $0 = \beta_1' = \beta_1 = 0$. 
Moreover, by \eqref{eq:m_assump_informal}, $\overline{\w_0}  = \M^{-1}\b + \Delta_0$, where $\Delta_0$ is a vector with $\|\Delta_0\|_{\M} \leq \epsilon_0 \|\M^{-1}\b\|_{\M}$. I.e. $\overline{\w_0} = \M^{-1}(\b + \M \Delta_0)$. So, for some $\b'$ satisfying
\begin{align}
\label{eq:how_close_is_b}
 \b'-\b = \M \Delta_0,
\end{align}
we have that, for $\w_0 = \M^{-1/2}\b'$,
\begin{align}
\label{eq:w0_overline_equal}
\overline{\w_0} = \M^{-1/2}\w_0. 
\end{align}
Next, we have that 
\begin{align*}
(z')^2 = \langle \b, \overline{\w_0}\rangle = \|\M^{-1/2}\b\|^2 + \langle \Delta_0, \b\rangle. 
\end{align*}
By Cauchy-Schwarz, $|\langle \Delta_0, \b\rangle| \leq \|\M^{-1/2}\b\|\|\M^{1/2}\Delta_0\|$ and we have that $\|\M^{1/2}\Delta_0\| = \|\Delta_0\|_{\M} \leq \epsilon_0 \|\M^{-1/2}\b\|_2$.
It follows that, as long as $\epsilon_0 \leq \mach$, $(z')^2 = (1+\Delta_1)\|\M^{-1/2}\b\|^2$ for some scalar $\Delta_1$ with $|\Delta_1| \leq \mach$, and therefore that:
\begin{align*}
z' = (1+\Delta_2)\|\M^{-1/2}\b\|
\end{align*}
for scalar $\Delta_2$ with $|\Delta_2| \leq \mach$. On Line 3 of \Cref{alg:lanczos_sym_precon}, Paige's analysis allows for $z \gets \|\w_0\|$ to be computed up to relative error $(1+\mach(n+2)/2)$ (see \cite{paige1976error}, equation 12). So, there is a choice of acceptable roundoff error for which $z' = z.$
Combined with \eqref{eq:w0_overline_equal}, we conclude that:
\begin{align}
\label{eq:q1_overline_equal}
\overline{\q_1} = \M^{-1/2}\q_1.
\end{align}
This establishes the second equation in \eqref{eq:3}, so to complete the base case analysis for $i=1$, we are left to address the third equation, i.e. that $\|\underline{\q_1} - \M^{1/2}\q_1\| \leq \mach \|\M^{1/2}\|$. We have from \eqref{eq:q1_overline_equal} that $\M^{1/2}\q_1 = \M\overline{\q_1} = \M\cdot (\M^{-1}\b + \Delta_0)/z'$. We have that $\underline{\q_1} = \b/z'$, so we conclude that:
\begin{align*}
\|\underline{\q_1} - \M^{1/2}\q_1\|_{\M^{-1}} = \|\M\Delta_0/z'\|_{\M^{-1}} = \|\Delta_0/z'\|_{\M}.
\end{align*}
As shown above, as long as $\mach \leq 1/2$, $z' \geq \frac{1}{2}\|\M^{-1/2}\b\|$, so we have that 
\begin{align*}
\|\underline{\q_1} - \M^{1/2}\q_1\|_{\M^{-1}} \leq \frac{2}{\|\M^{-1/2}\b\|} \|\Delta_0\|_{\M} \leq 2\epsilon_0.
\end{align*}
So, we have prove the third equation of \eqref{eq:3} as long as $\epsilon_0 \leq \frac{\mach}{2}.$

\medskip\noindent\textbf{Inductive Case.} We can now move onto the inductive case of \eqref{eq:3}. I.e., to proving the statement for all $i \geq 2$, assuming it holds for $j < i$.

We proceed with a line by line analysis of \Cref{alg:lanczos_left_precon,alg:lanczos_sym_precon}.

\smallskip\noindent\textbf{Line 6.}
Our first goal is to prove that for some choice of acceptable roundoff error in \Cref{alg:lanczos_sym_precon},
\begin{align}
\label{eq:underline_u_equal}
   \M^{-1/2} \underline{\u_i} = \u_i.
\end{align} 
Observe that 
\begin{align*}
\M^{-1/2}\underline{\u_i} &= \M^{-1/2} \A \overline{\q_i} - \beta_i' \M^{-1/2}\underline{\q_{i-1}}\\
 &= \M^{-1/2} \A \M^{-1/2}{\q_i} - \beta_i \M^{-1/2}\underline{\q_{i-1}}\\
  &= \M^{-1/2} \A \M^{-1/2}{\q_i} - \beta_i\q_{i-1} + \beta_i\Delta_3,
\end{align*}
where, using our inductive assumption, $\Delta_3 = \q_{i-1} - \M^{-1/2}\underline{\q_{i-1}}$ is a vector with Euclidean norm bounded by $\mach$.
Paige's analysis allows for $\u_i \leftarrow \M^{-1/2} \A \M^{-1/2}{\q_i} - \beta_i\q_{i-1}$ to be computed up to additive error bounded in norm by $\mach 2\|\beta_i\q_{i-1}\|$ (see \cite{paige1976error}, equation 8). 

Using that $\|\q_i\|  \geq 1/2$ as long as $\mach \leq \frac{1}{2(n+4)}$ (see \cite{paige1976error} equation 14) we thus have that there is a choice of allowable error in \Cref{alg:lanczos_sym_precon} such that $\M^{-1/2}\underline{\u_i} = \u_i$.
So, we have proven \eqref{eq:underline_u_equal} as desired.

\smallskip\noindent\textbf{Line 7.}
Next, by our inductive assumption and \eqref{eq:5}, we have that 
\begin{align*}
\langle \underline{\u_i},\overline{\q_i}\rangle = \langle \u_i,\q_i\rangle,
\end{align*}
so we immediately have that, even with no roundoff error in \Cref{alg:lanczos_sym_precon},
\begin{align}
\label{eq:5}
\alpha_i' = \alpha_i.
\end{align}
This establishes the first condition of \eqref{eq:3}.

\smallskip\noindent\textbf{Line 8.}
We want to prove that, for some choice of acceptable rounding error in \Cref{alg:lanczos_sym_precon},
\begin{align}
\label{overline_w_equality}
    \M^{1/2}\overline{\w_i} = \w_i.
\end{align}
By \Cref{eq:5} and \Cref{eq:m_assump_informal} we have:
\begin{align}
    \overline{\w_i} &= \M^{-1}\underline{\u_i} - \alpha_i\M^{-1}\underline{\q_i} + \Delta_4,
\end{align} 
where $\|\Delta_4\|_{\M}  \leq \epsilon_0 \|\M^{-1/2}(\underline{\u_i} - \alpha_i\underline{\q_i})\|$. 
Furthermore, by \eqref{eq:3}, we have that $\M^{-1/2}\underline{\q_i} = \q_{i} + \Delta_5$, where $\|\Delta_5\| \leq \mach$. Combining with \eqref{eq:underline_u_equal}, it follows that:
\begin{align}
    \M^{1/2}\overline{\w_i} &= \M^{-1/2}\underline{\u_i} - \alpha_i\M^{-1/2}\underline{\q_i} + \M^{1/2}\Delta_4 \\
    &= \u_i - \alpha_i \q_{i+1} + \alpha_i\Delta_5 + \M^{1/2}\Delta_4.
\end{align} 
Paige's analysis requires that $\w_i$ equals $\u_i - \alpha_i \q_i$ up to additive error with norm bounded by $\mach(\|\u_i\| + 2\alpha_i\|\q_i\|)$ (see \cite{paige1976error}, equation 8). So to establish \eqref{overline_w_equality}, we need to show that:
\begin{align}
\label{eq:delta4_delta5_norm_bound}
    \|\alpha_i\Delta_5 + \M^{1/2}\Delta_4\| \leq \mach(\|\u_i\| + 2\alpha_i\|\q_i\|).
\end{align}
Observe that, by triangle inequality,
\begin{align*}
    \|\M^{1/2} \Delta_4\| \leq \epsilon_0 \|\M^{-1/2}(\underline{\u_i} - \alpha_i\underline{\q_i})\| &= \epsilon_0 \|\u_i - \alpha_i\M^{-1/2}\underline{\q_i}\|\\
    &\leq \epsilon_0 \|\u_i - \alpha_i{\q_i}\| + \alpha_i\Delta_5.
\end{align*}
Combined with our bound on $\Delta_5$, we have:
\begin{align*}
 \|\alpha_i\Delta_5 + \M^{1/2}\Delta_4\| \leq \epsilon_0 \|\u_i\| + 2\alpha_i\mach.
\end{align*}
Using that $\|\q_i\| \in [1/2,2]$ as long as $\mach \leq \frac{1}{2(n+4)}$ (see \cite{paige1976error} equation 14), we see that \eqref{eq:delta4_delta5_norm_bound} holds as long as $\epsilon_0 \leq \mach$.
It follows that \eqref{overline_w_equality} holds.

We also claim that, directly from \eqref{overline_w_equality} and \eqref{eq:m_assump_informal},
\begin{align}
\label{underline_w_error}
    \|\w_i - \M^{-1/2}\underline{\w_i}\| = \|\overline{\w_i} - \M^{-1}\underline{\w_i}\|_{\M} \leq \epsilon_0 \|\M^{-1}\underline{\w_i}\|_{\M} = \epsilon_0 \|\M^{-1/2}\underline{\w_i}\|.
\end{align}

\smallskip\noindent\textbf{Line 9.}
By \eqref{underline_w_error} and \eqref{overline_w_equality}, we have that:
\begin{align*}
{\langle \underline{\w_i},\overline{\w_i}\rangle} = {\langle \M^{-1/2}\underline{\w_i},\M^{1/2}\overline{\w_i}\rangle}  = \langle \w_i,\w_i\rangle + \langle \Delta_6,\w_i\rangle,
\end{align*}
where $\|\Delta_6\| \leq  \epsilon_0 \|\M^{-1/2}\underline{\w_i}\|$.
Applying Cauchy-Schwarz, it follows that:
\begin{align*}
\langle \underline{\w_i},\overline{\w_i}\rangle = \langle {\w_i},{\w_i}\rangle + \Delta_7,
\end{align*}
where $\Delta_7$ is a scalar with $|\Delta_7| \leq \epsilon_0 \|\w_i\|\|\M^{-1/2}\underline{\w_i}\|$. As long as $\epsilon_0 \leq 1/2$, we can check from \eqref{underline_w_error} that $\|\M^{-1/2}\underline{\w_i}\| \leq 2 \|\w_i\|$.
It follows that $\langle \underline{\w_i},\overline{\w_i}\rangle = (1+\Delta_8)\langle {\w_i},{\w_i}\rangle$ where $|\Delta_8| \leq 2\epsilon_0$, and thus 
\begin{align}
\sqrt{\langle \underline{\w_i},\overline{\w_i}\rangle} = (1+\Delta_9)\sqrt{\langle {\w_i},{\w_i}\rangle},
\end{align}
again for $|\Delta_9| \leq 2\epsilon_0$. Paige's analysis assumes that $\sqrt{\langle \underline{\w_i},\overline{\w_i}\rangle}$ is computed up to multiplicative accuracy $(1+ \mach(n+2)/2)$ (see \cite{paige1976error} equation (12)), so as long as $ 2\epsilon_0 \leq \mach(n+2)/2$, then there is some choice of acceptable roundoff error such that, as desired.
\begin{align}
\label{eq:betas_equal}
    \beta_{i+1}' = \beta_{i+1},
\end{align}
This establishes the second condition of \eqref{eq:3}.

\smallskip\noindent\textbf{Line 10.}
By \eqref{eq:betas_equal}, the if condition in Line 9 evaluates to true in \Cref{alg:lanczos_left_precon} if and only if it evaluates to true in  \Cref{alg:lanczos_sym_precon}, so we can move onto Line 13. 

\smallskip\noindent\textbf{Line 13.}
We have immediately from \eqref{overline_w_equality} and \eqref{eq:betas_equal} that:
\begin{align}
    \label{eq:underline_qiplus_equal}
    \overline{\q_{i+1}} = \M^{-1/2} \q_{i+1}.
\end{align}
This proves our third condition in \eqref{eq:3} for $i+1$.
Additionally, from \eqref{underline_w_error} and \eqref{eq:betas_equal}, we have that:
\begin{align*}
\|\underline{\q_i} - \M^{1/2}{\q_i}\|_{\M^{-1}} =  \frac{1}{\beta_i}\|\underline{\w_i} - \M^{1/2}{\w_i}\|_{\M^{-1}} = \frac{1}{\beta_i}\|\M^{-1/2}\underline{\w_i} - {\w_i}\|_{\M^{-1}} 
    &\leq \frac{\epsilon_0}{\beta_i}\|\M^{-1/2}\underline{\w_i}\|,
\end{align*}
As before, we have that $\|\M^{-1/2}\underline{\w_i}\| \leq 2 \|\w_i\|$ and, from \cite{paige1976error} equation 12, as long as $\mach \leq \frac{1}{n+2}$, $\beta_{i+1} \geq \frac{1}{2}\|{\w_i}\|$. So, we conclude that
\begin{align*}
\|\underline{\q_i} - \M^{1/2}{\q_i}\|_{\M^{-1}} \leq 4 \epsilon_0
\end{align*}
which implies our fourth condition in \eqref{eq:3} as long as $4 \epsilon_0 \leq \mach$.

\medskip\noindent
\textbf{Completing the Proof.}
With \eqref{eq:3} in place, we are ready to prove our main result, \Cref{thm:stable_lanczos}. We have established that if \Cref{alg:lanczos_left_precon} is run with input $\b$ and function $\mathsf{SolveM}$ satisfying \eqref{eq:m_assump_informal} with $\epsilon_0 \leq \mach/4$, then it generates a tridiagonal matrix $\T'$ that is identical to the output of \Cref{alg:lanczos_sym_precon} run with input $\b'$ and a particular set of round-off errors that satisfy the condition of \Cref{f:more_general}. Moreover, the matrix $\overline{\Q}$ generated by \Cref{alg:lanczos_left_precon} is identical to $\M^{-1/2}\Q$, where $\Q$ is the matrix generated \Cref{alg:lanczos_sym_precon}. Since we have also establishes that $z =z'$, it follows that Line 16 of each algorithm returns and identical output. We conclude from \Cref{f:more_general} that, as long as $\mach = \poly(\epsilon/n\kappa_{\M})$, for $t= O\left(\sqrt{\kappa_{\M}} \log (\kappa_{\M}/\epsilon)\right)$, the vector $\x_t'$ returned by \Cref{alg:lanczos_left_precon} satisfies:
\begin{align}
\label{eq:with_b_prime}
\|\x_t' - \A^{-1}\b'\|_{\A} \leq \epsilon \|\A^{-1}\b'\|_{\A}.
\end{align}
Recalling that $\b' = \b + \M\Delta_0$ for a vector $\Delta_0$ with $\|\M^{1/2}\Delta_0\| \leq \epsilon_0 \|\M^{-1/2}\b\|$, we further have that:
\begin{align*}
\|\A^{-1}\b - \A^{-1}\b'\|_{\A} = \|\A^{-1/2}\b - \A^{-1/2}\b'\| &= \|\A^{-1/2}\M\Delta_0\| \\
&\leq \epsilon_0|\A^{-1/2}\M^{1/2}\|\|\M^{-1/2}\b\|\\
&\leq \epsilon_0|\A^{-1/2}\M^{1/2}\|\|\M^{-1/2}\A^{1/2}\|\A^{-1/2}\b\|\\ &= \epsilon_0 \sqrt{\kappa_{\M}}\|\A^{-1}\b\|_{\A} 
\end{align*}
Combining with \eqref{eq:with_b_prime} via triangle inequality, we have that, as long as $\epsilon_0 \leq \frac{\epsilon}{\sqrt{\kappa_{\M}}}$, 
\begin{align*}
\|\x_t' - \A^{-1}\b\|_{\A} \leq 3\epsilon \|\A^{-1}\b\|_{\A}.
\end{align*}
Adjusting the constant on $\epsilon$ proves \Cref{thm:stable_lanczos}. 
\end{proof}

\section{Lower Bound for Linear Systems with $k$ Large Singular Values}
\label{s:lower}

In this section, we give a lower bound for the time complexity of solving linear systems with $k$ large singular values.
\begin{theorem}\label{t:lower}
    Assuming that the time complexity of solving an arbitrary $n\times n$ linear system  to precision $\epsilon=1/\poly(n)$ is $\Omega(n^\omega)$, the time complexity of solving an $n\times n$ linear system $\A\x=\b$ such that $\sigma_{k+1}(\A)/\sigma_{\min}(\A) = O(1)$ is at least $\Omega( n^2 + k^\omega)$.
\end{theorem}
\begin{proof}
    The proof will be divided into two parts. First, we will give an information theoretic lower bound showing that to solve a linear system with $k$ large singular values, one must read at least $\Omega(n^2)$ entries of the matrix $\A$. To do this, consider the following class of $n\times n$ matrices parameterized by an index pair $(i,j)$, where $\e_i$ denotes the $i$th $n$-dimensional standard basis vector:
    \begin{align*}
        \mathcal A = \Big\{ \I_n + \e_i\e_j^\top - \e_j\e_i^\top\quad:\quad i\neq j\Big\}.
    \end{align*}
    Matrix $\A\in\mathcal A$ is a matrix with identity on the diagonal and then two off-diagonal entries, $1$ and $-1$, symmetrically across from each other, and then zeros everywhere else. We then choose an all-ones vector $\b=[1,...,1]^\top\in\R^n$ and consider such a linear system:
    \begin{align*}
        \begin{bmatrix}
            1&~~&~~&~~&~~&~~\\
            ~~&1&&&1&\\
            &&1&&&\\
            &&&\ddots&&\\
            &-1&&&1&\\
            &&&&&1
        \end{bmatrix}
        \begin{bmatrix}x_1 \\x_2 \\ x_3\\\vdots \\ x_{n-1}\\x_n \end{bmatrix}=
        \begin{bmatrix}1\\1\\ 1\\ \vdots\\ 1\\1\end{bmatrix}
    \end{align*}
    Note that this linear system can be written more concisely as satisfying the following equations:
    \begin{align*}
        \text{for } l\neq i,j,\quad x_l &= 1,\\
        x_i+x_j &= 1,\\
        -x_i+x_j &= 1,
    \end{align*}
    which implies that the solution vector $\x$ is an all-ones vector except for $x_i=0$. To find the singular values associated with the matrix $\A$, observe that:
    \begin{align*}
        \A^\top\A = (\I_n + \e_i\e_j^\top - \e_j\e_i^\top)^\top(\I_n + \e_i\e_j^\top - \e_j\e_i^\top) = \I_n + \e_i\e_i^\top+\e_j\e_j^\top.
    \end{align*}
    Since $\A^\top\A$ is diagonal, its eigenvalues can be found immediately, which gives us the singular values of $\A$ (in decreasing order): $\sigma_1=\sqrt 2$, $\sigma_2=\sqrt 2$ and $\sigma_3=...=\sigma_n=1$. Thus, $\A$ is well-conditioned and it satisfies the $k$ large singular values property for any $k$. Suppose that the matrix $\A$ is given to us in an $n\times n$ two-dimensional array. How long does it take to find either of the non-zero off-diagonal entries of $\A$? Naturally, this will require in the worst-case reading $\Omega(n^2)$ entries, since we have no information  about the index pair $1\leq i<j\leq n$. On the other hand, if we were able to solve this linear system in time $o(n^2)$, obtaining (an approximation of) vector $\x$, then we could find $i,j$ pair in an additional $O(n)$ time by finding the entry of $\x$ that is not equal 1. This shows, by contradiction, that the time complexity of solving such a linear system must be $\Omega(n^2)$. 
    
    We note that, while our construction gives a sparse matrix, which may be stored more efficiently in a column-row-value format, this does not circumvent our lower bound, since we can also consider the above matrix $\A$ that has been distorted entry-wise with, say, random $\pm1/\poly(n)$ values, so that it is no longer sparse, but this distortion will not meaningfully affect its condition number or the linear system solution. Thus, the above lower bound applies also for such dense matrices.

    In the second part of the proof, we observe that to solve an $n\times n$ linear system with $k$ large singular values, one must (effectively) solve an arbitrary $k\times k$ linear system. Namely, consider an arbitrary $k\times k$ rank $k$ matrix $\M$ and a vector $\c\in\R^k$, and consider the following $n\times n$ linear system that can be written in a block-diagonal form:
    \begin{align*}
        \begin{bmatrix}
            \M&\\
            &\sigma_{\min}(\M)\cdot\I_{n-k}
        \end{bmatrix}\begin{bmatrix}
            \x\\
            \y
        \end{bmatrix}
        =
        \begin{bmatrix}
            \c\\
            \mathbf{0}_{n-k}
        \end{bmatrix},
    \end{align*}
    where we are solving for the vectors $\x$ and $\y$. Note that the singular values of this system are simply the singular values of $\M$, with its smallest singular value additionally repeated $n-k$ times. Thus, $\sigma_{k+1}/\sigma_{\min}=1$ for this system, and it satisfies the $k$ large singular values property. Also, note that the solution of this system is $\x=\M^{-1}\c$ and $\y=\mathbf{0}$. Now, if we can solve this system in time $o(k^\omega)$, then, by simply reading off the first $k$ entries of the solution vector, we obtain the solution to the $k\times k$ system $\M\x=\c$. This concludes the proof of the lower bound.
\end{proof}

\section*{Acknowledgements}
MD would like to acknowledge NSF CAREER for partial support. CM was partially supported by NSF Award (2045590).

\newpage
\bibliographystyle{alpha}
\bibliography{ref.bib,ref_old.bib}

\end{document}